\documentclass[11pt,letterpaper]{article}
\usepackage[letterpaper,margin=1in]{geometry}
\usepackage{standalone}
\usepackage{etoolbox}
\usepackage{amsthm}
\usepackage{amsmath}
\usepackage{amssymb}
\usepackage{thmtools}
\usepackage{thm-restate}
\usepackage{times}
\usepackage{mathtools}
\usepackage{stmaryrd}
\usepackage{placeins}
\usepackage[table]{xcolor}

\definecolor{linkcol}{rgb}{0,0,0.38}
\definecolor{urlcol}{rgb}{0.1,0.35,0}
\definecolor{citecol}{RGB}{50,125,70}

\usepackage[pdfpagelabels,bookmarks=false,hyperfootnotes=false]{hyperref}
\hypersetup{colorlinks, linkcolor=blue, citecolor=citecol, urlcolor=blue}

\usepackage[utf8]{inputenc}
\usepackage[
backend=biber,
style=alphabetic,
citestyle=alphabetic,
maxalphanames=4,
maxcitenames=99,
mincitenames=98,
maxbibnames=99,
giveninits=true,
]{biblatex}
\usepackage[nameinlink]{cleveref}
\usepackage{lmodern}
\usepackage{bbm}
\usepackage{thm-restate}

\newtheoremstyle{light} %
    {\topsep}                    %
    {\topsep}                    %
    {\itshape}                   %
    {}                           %
    {\scshape}                   %
    {.}                          %
    {.5em}                       %
    {}  %

\newtheorem{theorem}{Theorem}[section]
\newtheorem{lemma}[theorem]{Lemma}

\newtheorem{definition}[theorem]{Definition}

\theoremstyle{light}
\newtheorem{claiminproof}[theorem]{Claim}
\makeatletter
\if@cref@capitalise
\crefname{claiminproof}{Claim}{Claims}
\else
\crefname{claiminproof}{claim}{claims}
\fi

\if@cref@capitalise
\crefname{algocf}{Algorithm}{Algorithms}
\else
\crefname{algocf}{algorithm}{algorithms}
\fi

\if@cref@capitalise
\crefname{conjecture}{Conjecture}{Conjectures}
\else
\crefname{conjecture}{conjecture}{conjectures}
\fi

\if@cref@capitalise
\crefname{thm}{Theorem}{Theorems}
\else
\crefname{thm}{theorem}{theorems}
\fi

\if@cref@capitalise
\crefname{lem}{Lemma}{Lemmas}
\else
\crefname{lem}{lemma}{lemmas}
\fi

\makeatother

\usepackage{url}
\usepackage{array}
\usepackage{setspace}
\usepackage{wrapfig}
\usepackage{upgreek}
\usepackage{complexity}
\usepackage[inline]{enumitem}
\usepackage{moreenum}
\usepackage[super]{nth}
\usepackage{bm}
\usepackage{xspace}
\usepackage{ifthen}
\usepackage{comment}
\usepackage[immediate]{silence}
\usepackage{sectsty}
\DeactivateWarningFilters[tmp]
\allsectionsfont{\boldmath}

\setlist[enumerate]{nosep,topsep=0.1em}
\setlist[enumerate,1]{label=(\roman*), leftmargin=2.2em}
\setlist[itemize]{nosep,topsep=0.3em}
\usepackage[textsize=footnotesize, color=blue!30!white]{todonotes}
\usepackage{float}
\usepackage{graphicx}
\usepackage[margin=10pt,font=small,labelfont=bf,skip=-5pt]{caption}
\usepackage{subcaption}
\usepackage[vlined,ruled,algo2e]{algorithm2e}
\SetAlgoSkip{bigskip}
\SetAlgoInsideSkip{smallskip}

\let\truehypersetup\hypersetup
\renewcommand\hypersetup[1]{}
\usepackage{bigfoot}
\let\hypersetup\truehypersetup
\newcommand{\Zp}{\ensuremath{\mathbb{Z}_{\geq 0}}}
\newcommand{\Z}{\ensuremath{\mathbb{Z}}}
\renewcommand{\R}{\ensuremath{\mathbb{R}}}
\newcommand{\Rp}{\ensuremath{\mathbb{R}_{\geq 0}}}

\newcommand{\Dmax}{\ensuremath{d_{{\max}}}}
\newcommand{\Iscr}{\mathcal{I}}
\newcommand{\Icut}{\mathcal{I}^{\text{cut}}}
\newcommand{\Isplit}{\mathcal{I}^{\text{split}}}
\newcommand{\Qscr}{\mathcal{Q}}

\newcommand{\etight}{e^*}
\newcommand{\Etight}{E^*}

\newcommand{\Einner}{E_{\text{inner}}}
\newcommand{\Eouter}{E_{\text{outer}}}
\newcommand{\Pinner}{\ensuremath{P^{\mathrm{in}}}}
\newcommand{\Pouter}{\ensuremath{P^{\mathrm{out}}}}
\renewcommand{\NP}{\text{NP}}

\newcommand{\nice}[1]{#1-nice}

\definecolor{green}{rgb}{0.4,0.85,0.6}
 
\makeatletter
\def\@fnsymbol#1{\ensuremath{\ifcase#1\or *\or %
\ddagger\or
    \mathsection\or \mathparagraph\or \|\or **\or \dagger\dagger
    \or \ddagger\ddagger \else\@ctrerr\fi}}
\makeatother

\title{Pinning on Tight Cuts: Improved Algorithm and Bounds for Unsplittable Multicommodity Flows in Outerplanar Graphs\footnote{An extended abstract of this paper is to appear in the proceedings of the 53rd EATCS International Colloquium on Automata, Languages, and Programming (ICALP 2026), Track A.}}

\author{
David {Alem\'an Espinosa}\thanks{
Department of Combinatorics and Optimization, University of Waterloo, Canada.
Partially supported by the NSERC Discovery grant 2024-04532.
Email: \href{mailto:dalemanespinosa@uwaterloo.ca}%
{dalemanespinosa@uwaterloo.ca}. 
}
\and
Niklas Schlomberg\thanks{
Research Institute for Discrete Mathematics and Hausdorff Center for Mathematics, University of Bonn, Germany. Partially supported by the SNSF grant 200021-236706.
Email: \href{mailto:niklas.schlomberg@supsi.ch}%
{niklas.schlomberg@supsi.ch}.}
}

\date{}

\begin{document}
\maketitle
\begin{abstract}
The multicommodity flow problem in an undirected capacitated graph $G$ is specified by a set of source-sink pairs with nonnegative demands.
A flow is \emph{feasible} if it routes all demands without exceeding the edge capacities,
and it is \emph{unsplittable} if it routes each demand along a single path.

Let $\alpha$ be the \emph{smallest value} such that the existence of a feasible flow implies the existence of an unsplittable flow that exceeds the edge capacities by at most $+\,\alpha\,\Dmax$, where $\Dmax$ is the maximum demand value. 
Schrijver, Seymour, and Winkler showed that $\alpha\in\left[1.01,\,1.5\right]$ if $G$ is a cycle.
These bounds were ultimately improved to $\alpha\in\left[1.1,\,1.3\right]$ by Skutella and D\"aubel.
Recently,
Alem\'an Espinosa and Kumar extended this constant upper bound to the broader class of outerplanar graphs, and showed that if $G$ is outerplanar then $\alpha\le3.6$.

We show that $\alpha\in\left[\tfrac{4}{3},2\right]$ if $G$ is outerplanar.
We introduce a novel technique that considers the global parameters of the instance,
and that may be useful in other (more general) settings
where the cut-condition is sufficient, or nearly sufficient, for the existence of a feasible flow.
\end{abstract}

\section{Introduction}
Network flows are a fundamental and extensively studied class of problems in  combinatorial optimization; see, e.g.,~\cite{AhujaMagnantiOrlin1993}. 
The {\em multicommodity flow} problem involves routing multiple distinct commodities through a shared network.
An instance, which we denote by the tuple $(G,u,H,d)$, is given by an undirected \emph{supply graph} $G=(V,E(G))$ with edge capacities $u:E(G)\to\Rp$,
and a collection of source-sink pairs $\{s_i,t_i\}$ in $V$ with associated nonnegative demands $d(s_i,t_i)$.
It will be convenient to think of the source-sink pairs as forming the edges of a \emph{demand graph} $H=(V,E(H))$.
The goal is to compute a flow in $G$ that simultaneously routes all demands and respects the edge capacities, or to certify that no such flow exists.
In the standard (\emph{fractional}) version of the problem, 
the demand between a source-sink pair may be split across multiple paths,
and this flexibility enables the precise formulation of the problem as a linear program that can be solved efficiently.
This model has been studied for
almost 70 years~\cite{ford1958suggested,gomory1961multi}.

\vspace{1ex}
\noindent\textbf{Unsplittable multicommodity flows.}
Many practical applications often require the flow between each source-sink pair to be routed along a single path.
Such a flow is called \emph{unsplittable}.\footnote{In this work, whenever we refer to an unsplittable flow, we also assume that it routes \emph{all the demands}.}

This additional requirement gives rise to a significantly harder and less well understood problem.
Various fundamental $\NP$-complete problems in combinatorial optimization such as
bin packing, partitioning, and makespan minimization on identical and
related machines, can be reduced to \emph{unsplittable multicommodity flow} problems in networks consisting of only two nodes, which are the source-sink nodes of all demands, and parallel edges between them~\cite{kleinberg1996single}.
Moreover, already for very small feasible instances, any unsplittable flow 
violates some edge capacity by at least an additive amount of $\Dmax$,
where $\Dmax=\max_{i}d(s_i,t_i)$ denotes the maximum demand value (see Figure~\ref{fig:example}).
Furthermore, strong impossibility and/or hardness results imply that even with unit demands and unit capacities, and even when a feasible multicommodity flow exists, finding an \emph{unsplittable flow} with only a small violation of edge capacities may be impossible or computationally intractable; see, e.g., the (updated) survey by Kolliopoulos~\cite{kolliopoulos2018}.

Therefore, a natural question is to study the conditions under which the existence of a feasible flow implies the existence of an unsplittable flow that incurs only a small violation of edge capacities.
Extensive work studies this question for special instances~\cite{dinitz1999single,schrijver1998ringloading,skutella2002approximating,Martens2007,skutella2016ringloading,daubel2019ringloadingBestUpperbound,morell2022single,shapleyshmoys,traub2024single,aleman2025unsplittable,majthoub2025integer,aleman26,swamy2026unsplittable}.

In the \emph{single-source} setting (i.e., when the demand graph $H$ is a star),
the most prominent result of this kind is due to Dinitz, Garg, and Goemans~\cite{dinitz1999single},
who proved that any feasible (fractional) flow can be converted into an unsplittable flow that violates the edge capacities by at most $\Dmax$, the value of the maximum demand.\footnote{This result also holds for \emph{directed single-source} instances.}

However, in the multicommodity setting very few results of this kind are known for unsplittable flows~\cite{schrijver1998ringloading,shapleyshmoys,aleman2025unsplittable,majthoub2025integer,aleman26}.
A prominent example is the case where the supply graph $G$ is a cycle and the demand graph $H$ is arbitrary, 
which is known as the \emph{ring-loading} problem, first studied by Cosares and Saniee~\cite{cosares1994optimization}.
In a classical result, Schrijver, Seymour, and Winkler~\cite{schrijver1998ringloading} proved that any feasible (fractional) flow for a ring-loading instance can be converted into an unsplittable flow that exceeds the edge capacities by at most an additive amount of $1.5\,\Dmax$;
the best known bound, due to D\"aubel~\cite{daubel2019ringloadingBestUpperbound}, is $1.3\,\Dmax$.

Shapley and Shmoys~\cite{shapleyshmoys} considered the much broader class of instances in which $G$ is \emph{outerplanar} and $H$ is arbitrary.
Recall that $G$ is outerplanar if it admits a planar embedding in which all vertices lie on the unbounded face.
They showed that any feasible flow can be converted into an unsplittable flow that exceeds capacities by at most $O(\log f)\,\Dmax$, where $f$ is the number of faces of $G$.
This was recently improved by Alem\'an Espinosa and Kumar~\cite{aleman2025unsplittable}, who gave an additive $3.6\,\Dmax$ congestion bound.

We show that an additive congestion of $2\Dmax$ can be attained in outerplanar graphs.
In terms of lower bounds on the additive violation needed, the only known result is a lower bound of $1.1\,\Dmax$ by Skutella~\cite{skutella2016ringloading} for the ring-loading problem.
We show that the additive congestion has to be at least $\tfrac{4}{3}\,\Dmax$ for outerplanar graphs.
We remark that outerplanar graphs are the only known nontrivial class of supply graphs for which the existence of a feasible flow guarantees the existence of an unsplittable flow exceeding edge capacities by at most $O(\Dmax)$, regardless of the demand graph $H$.

\begin{figure}
    \centering
    \includegraphics[width=0.9\textwidth]{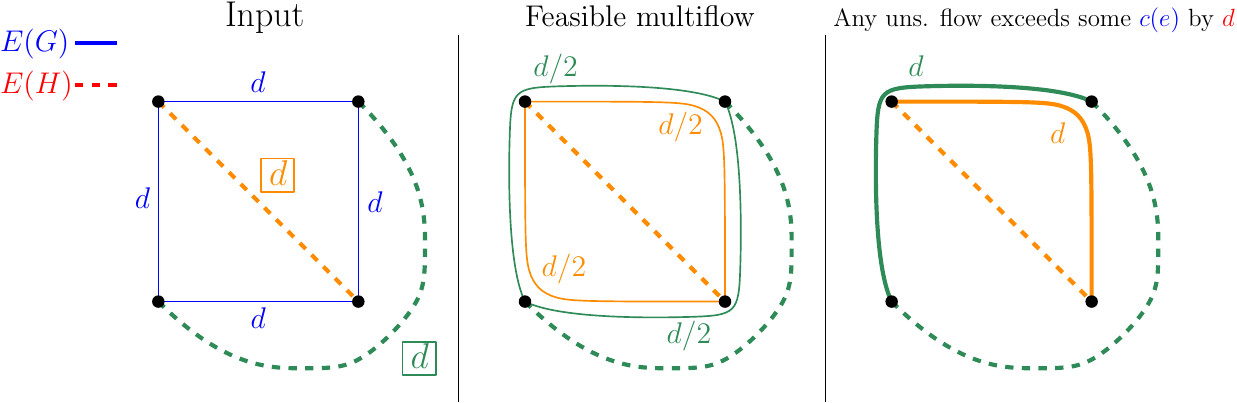}

    \vspace{5mm}
    \caption{An example of a small feasible instance, where every unsplittable flow must necessarily exceed the capacity of some edge by at least $\Dmax$.}\label{fig:example}
\end{figure}

\subsection{Our Contributions}
In this work, we significantly improve the state-of-the-art capacity-violation bounds for the existence of an unsplittable flow in outerplanar graphs.

Let $\alpha_O>0$ be the smallest value such that every feasible instance
with an outerplanar supply graph $G$ admits an unsplittable flow that exceeds the edge capacities by at most an additive amount of $\alpha_O\,\Dmax$.
Similarly, let $\alpha_{RL}$ denote the smallest associated value for \emph{ring-loading} instances (the example in Figure~\ref{fig:example} implies that $\alpha_{RL}\ge 1$).

Prior work showed that $1.01\le\alpha_{RL}\le1.5$~\cite{schrijver1998ringloading}, which was ultimately improved to $\alpha_{RL}\leq 1.3$~\cite{daubel2019ringloadingBestUpperbound}, and the lower bound $\alpha_{RL}\geq 1.1$~\cite{skutella2016ringloading}. Since every cycle is an outerplanar graph, we have $\alpha_O\geq\alpha_{RL}$, and prior to our work, it was not known that the parameters $\alpha_O$ and $\alpha_{RL}$ are necessarily different.

Our first main result establishes that $\alpha_O\le 2$, improving on the bound of Alem\'an Espinosa and Kumar~\cite{aleman2025unsplittable}, who showed that $\alpha_O \le 2\alpha_{RL}+1 \in [3.2,3.6]$.
Our proof uses a more robust divide-and-conquer framework that treats the parameters of the instance more globally.
We hope that this approach will be useful for obtaining similar congestion bounds in other settings where the cut condition is sufficient, or approximately sufficient, for feasibility.

\begin{restatable}{theorem}{theoremmain}\label{theorem:main}
Let $G$ be an outerplanar graph.
If an instance $(G,u,H,d)$ is feasible, or equivalently, if it satisfies the cut-condition, then there exists an unsplittable flow $y$ such that
\[y(e)\le u(e)+2\Dmax\quad \text{for all }e\in E(G).\]
Moreover, $y$ can be computed in polynomial time.
\end{restatable}

We complement this result by showing 
that $\alpha_{O}\geq \tfrac{4}{3}$.

\begin{restatable}{theorem}{theoremlower}\label{theorem:lower}
For any $\varepsilon>0$ there exists a feasible instance $(G,u,H,d)$ on an outerplanar supply graph $G$, such that for any unsplittable flow $y$, there exists a supply edge $e\in E(G)$ with
\[y(e)>u(e)+\left(\tfrac{4}{3}-\varepsilon\right)\Dmax.\]
\end{restatable}

Recall that previously, the only lower bound known on $\alpha_O$ was that $\alpha_O\geq\alpha_{RL}\geq 1.1$. Our improved lower bound in Theorem~\ref{theorem:lower}, together with the fact that $\alpha_{RL}\leq 1.3$~\cite{daubel2019ringloadingBestUpperbound}, shows that the $\alpha$-parameters for ring-loading and outerplanar graphs are different, and concretely shows that outerplanar graphs are more difficult settings than ring-loading instances for multicommodity unsplittable flows.

\subsection{Related Work}
Very recently, Alem\'an Espinosa, Kumar, Poremba, and Shepherd~\cite{aleman26} showed that in the \emph{fully planar setting}, defined by instances in which $G+H$ is planar, any feasible flow can be converted into an unsplittable flow that violates the edge capacities by at most $2\Dmax$. 
They also showed that the additive congestion has to be at least $1.5\,\Dmax$ in this setting.
The authors in~\cite{aleman26} also showed that when $G$ is \emph{series-parallel} (i.e., $G$ does not contain a $K_4$ minor), any feasible flow can be converted into an unsplittable flow such that the flow on any edge is at most twice its capacity plus $7.2\,\Dmax$.

The directed single-source setting has been extensively studied; see, e.g., \cite{skutella2002approximating,Martens2007,morell2022single,traub2024single,swamy2026unsplittable}, including work on minimum-cost variants. Recently, Majthoub Almoghrabi, Skutella, and Warode~\cite{majthoub2025integer} studied a multicommodity setting in a special class of directed series-parallel graphs and showed that any fractional flow can be converted into an unsplittable flow of no larger cost, such that the flow values on each arc differ by at most $\Dmax$.\footnote{
We note that, as mentioned in~\cite{majthoub2025integer} (see footnote on page 429 in~\cite{majthoub2025integer}), that their results do not extend to the undirected (outerplanar) setting we consider here, nor the other way around.}

\subsection{The Cut-Condition}
Interestingly, and perhaps not surprisingly,
all of these settings have the common feature that one can identify a simple structural property, called the {\em cut-condition} that is sufficient, or nearly-sufficient, to guarantee the existence of a feasible flow.
This condition requires that, for every cut, the total demand separated by the cut is at most the total capacity of the supply edges crossing it.
The cut-condition is always necessary for feasibility of a multicommodity flow instance (see, e.g.,~\cite{schrijver2003combinatorial}), but it is not always sufficient, even for small graphs such as $G=K_{2,3}$~\cite{okamura1981multicommodity}.
It is well known that the cut-condition is sufficient for the following instances:

\vspace{1ex}
\begin{enumerate}[label=(\roman*)]
\item (Single-source setting) $G$ is arbitrary and $H$ is a star~\cite{ford1956maximal}.
\item (Okamura-Seymour setting)
$G$ is planar and all endpoints of the demand edges lie on a common face of $G$~\cite{okamura1981multicommodity}.\label{instance:OS}
\item (Outerplanar setting) $G$ is outerplanar and $H$ is arbitrary (special case of (ii)).
\item (Fully planar setting) $G+H=(V,E(G)\cup E(H)\,)$ is planar~\cite{seymour1981odd}.
\end{enumerate}
\vspace{1ex}

Alem\'an Espinosa and Kumar~\cite{aleman2025unsplittable} posed as an open problem whether every feasible Okamura--Seymour instance admits an unsplittable flow with additive congestion $O(\Dmax)$.
Alem\'an Espinosa, Kumar, Poremba, and Shepherd~\cite{aleman26} later posed a more general question relating unsplittable flows and the flow-cut gap.

Given a (supply graph, demand graph) pair $(G,H)$, the \emph{flow-cut gap} of $(G,H)$ is the smallest value $\beta\ge 1$ such that, for all choices of capacities and demands, if every cut has capacity at least $\beta$ times the total demand it separates, then a feasible flow exists.
For example, $\beta=1$ for the classes of instances (i)--(iv).
It is known that $\beta=2$ for instances in which $G$ is series-parallel and $H$ is arbitrary~\cite{chakrabarti2008embeddings,lee2010coarse}.
For instances in which $G$ is planar it is only known that $\beta\in[2,\,O(\sqrt{\log n}\,)\,]$~\cite{Rao1999PlanarEmbedding,lee2010coarse}.
In contrast, for general graphs the flow-cut gap is known to be $\beta=\Theta(\log k)$, where $k$ is the number of source-sink pairs~\cite{linial1995geometry,aumann1998log}.
This motivates the natural question, posed in~\cite{aleman26}, of whether a flow-cut gap of $\beta$ for a class of instances implies the existence of an unsplittable flow that can be routed with capacities $O(\beta)\,u+O(\Dmax)$, or even $f(\beta)\,u+O(\Dmax)$ for some function $f:\Rp\to \Rp$.\footnote{An analogous question was previously posed by Chekuri, Shepherd, and Weibel~\cite{chekuri2013flow} in the context of \emph{integral multicommodity flows}, i.e., the case of unit demands. They answer this question affirmatively for various classes of instances, including settings in which $G$ is series-parallel or $k$-outerplanar.}

The sufficiency of the cut-condition for the outerplanar setting is used crucially in~\cite{aleman2025unsplittable} and also in this work.
Its approximate sufficiency for series-parallel graphs is used in~\cite{aleman26}, where the framework in~\cite{aleman2025unsplittable} is carefully adapted to handle this setting.
We build on the ideas in~\cite{aleman2025unsplittable} and develop a more robust divide-and-conquer framework that treats the instance parameters more globally.
We hope that this approach can be extended beyond the outerplanar setting, for example, to the significantly more intricate Okamura--Seymour setting $(ii)$, which is perhaps the most natural (intermediate) generalization.

\subsection{Technical Contributions and Overview}

Pinning a demand edge $h=(s,t)\in E(H)$ is a standard technique in (unsplittable) multicommodity flow problems.
It subdivides $h$ into a sequence $(s,v_1),(v_1,v_2),\ldots,(v_\ell,t)$ of demand edges, each with the same demand value, with the goal of transforming the instance into a more structured one.
However, pinning can destroy feasibility: the resulting instance may not even admit a fractional flow.
The algorithms in~\cite{dinitz1999single}, \cite{aleman2025unsplittable}, \cite{majthoub2025integer}, and~\cite{aleman26} (for series-parallel graphs), perform pinnings in a controlled way so that an unsplittable flow exists, while incurring only an $O(\Dmax)$ additive capacity violation.

We first provide a high-level overview of the algorithm in~\cite{aleman2025unsplittable} for outerplanar instances, which proceeds in two phases.
In the first phase, it performs pinnings until every demand edge becomes parallel to a supply edge. They show that the cut-condition, and hence feasibility, can be preserved throughout this phase by increasing capacities by at most $2\alpha_{RL}\,\Dmax$, where $\alpha_{RL}\le 1.3$~\cite{daubel2019ringloadingBestUpperbound} is the ring-loading parameter. In the second phase, they compute an unsplittable flow greedily, exceeding capacities by an additional amount of $\Dmax$.
We say that
an ear of $G$ is a path whose internal vertices have degree two, ignoring parallel edges.
In each iteration of the first phase, their algorithm selects an ear $P=v_1,\ldots,v_\ell$ (with $\{v_1,v_\ell\}\in E(G)$),
and considers the set $H_P\subseteq E(H)$ of demand edges with at least one endpoint in $P$. Each demand $(v_i,w)\in H_P$ is then pinned along a subpath of $P$.
The key observation in~\cite{aleman2025unsplittable} is that all demands in $H_P$ can be pinned via a reduction to ring-loading,
which preserves the cut-condition in the resulting instance,
at the expense of increasing capacities on $E(P)+\{v_1,v_\ell\}$ by an additive amount of $\alpha_{RL}\,\Dmax$.
After the pinnings along $P$ are made, a feasible flow is computed, and then each demand $(s,t)\in E(H)\cap E(G)$ and its corresponding 
$s$-$t$ flow are removed from the first phase.
After this removal, $P$ can be treated as an edge that is parallel to $\{v_1,v_\ell\}$ for the remaining demands $E(H)\setminus H_P$.
A similar ring-loading reduction approach was implemented in~\cite{aleman26} in order to contract the ears of series-parallel graphs,
where the goal of the first phase is to pin the demands so that the resulting instance becomes fully planar (i.e., $G+H$ is planar).

The main drawback of this procedure is that we cannot remove the ear $P$ after an iteration because the demands in $E(H)\setminus H_P$ may still require the residual capacity of $P$.
This is not an issue in the outerplanar setting,
but it becomes critical in, for example, an Okamura-Seymour instance $(G,u,H,d)$.
For example, suppose that both endpoints of every demand edge lie on the unbounded face of $G$.
Not being able to remove $P$ implies that one cannot pin the demands by using the inner nodes of $G$.
If one makes such a pinning, then the resulting instance is no longer an Okamura-Seymour instance.
Therefore, one cannot get access to these inner nodes by relying on such a ring-loading reduction.

We now give an overview of our algorithm.
Our main technical insight is that, instead of pinning on an ear,
we can pin on the vertices incident to the cut edges $\delta_G(S)$ of a tight cut $(S,\,V\setminus S)$, that is, a cut satisfying $u(\delta_G(S))=d(\delta_H(S)).$
This allows us to decompose the outerplanar instance $\Iscr$ into two smaller instances, which can then be solved recursively.

More concretely, tightness of the cut implies that any feasible flow for $\Iscr$ must route every demand with both endpoints on the same side of the cut entirely within that side.
We start by computing a feasible flow $x=(x_h)_{h\in E(H)}$.
We then remove every demand $h\in E(H[S])$ together with its $h$-flow $x_h$. That is, we reduce the capacity of each edge $e\in E(G[S])$ by $\sum_{h\in E(H[S])} x_h(e).$

Our first task is to construct an unsplittable flow for the demands that remain in the instance.
To this end, we first compute a simplified \emph{cut instance} $\Icut$, obtained from $\Iscr$ by removing the edges of the induced graph $G[S]$ that are not incident to the unbounded face of $G$.
In order to guarantee the feasibility of $\Icut$, we transfer the (residual) capacity of the deleted edges to the remaining edges of $G[S]$.

We solve the (edgewise smaller) instance $\Icut$ recursively, i.e., we find an unsplittable flow $y^\prime$ for $\Icut$ that exceeds the edge capacities by at most $2 \Dmax$.
For the edges of the unbounded face, we get a stronger bound of $+\frac{3}{2} \, \Dmax$.
The flow $y^\prime$ already determines our unsplittable flow for the initial instance $\Iscr$ on all edges outside $G[S]$, namely on $E(G)\setminus E(G[S])=E(G[V-S])\cup\delta_G(S)$.

Our second task is to extend this partial unsplittable flow to all of $G$.
To this end, we construct a \emph{split instance} $\Isplit$ whose supply graph is $G[S]$.
We first add back the removed demands $E(H[S])$.
Then, we use $y'$ to pin every demand $(s,t)\in E(H)\setminus E(H[S])$ of $\Icut$ on the vertices of $S$ incident to the cut  $\delta_G(S)$.
Suppose w.l.o.g. that $s\in V\setminus S$, and consider the path $P$ on which $y'$ routes $(s,t)$. 
Suppose that $P=P_1Q_1\hdots P_iQ_i\hdots$, where the $Q_i$'s
denote the maximal subpaths of $P$ lying inside $S$ (the last subpath of this sequence is some $Q_\ell$ if $t\in S$).
We fix the subpaths $P_i$ in our final unsplittable $s$-$t$ flow for $\Iscr$.
Each subpath $Q_i$ induces a new demand in $\Isplit$ between its endpoints.
We show that, since $y^\prime$ does not exceed capacities by too much, $\Isplit$ becomes feasible if we increase the capacities of the edges $\Etight \subseteq E(G[S])$ lying on the inner faces of $G$ that are incident to $\delta_G(S)$ by $\frac{3}{2}\,\Dmax$. All edges in $E(G[S])\setminus \Etight$ retain their original capacity in $\Iscr$.
This edge set $\Etight$ contains at most one edge per block of $G[S]$.

We again use recursion to solve the (smaller) split instance, which yields an unsplittable flow $\widehat y$ that can ``fill the gaps'' for $y^\prime$.
Our final unsplittable flow $y$ will be equal to $\widehat y$ on $G[S]$ and to $y^\prime$ outside $G[S]$.
For the edges $\Etight$ we will guarantee an even stronger bound of $+\frac{1}{2} \, \Dmax$ on the capacity violation of $\widehat y$, which will yield the capacity bounds needed for $y$.

The only case where our recursion step fails is if the tight cut $\delta_G(S)$ already contains all inner edges of $G$, because then the cut instance $\Icut$ will not be smaller than the original instance $\Iscr$.
In this case, the inner faces of $G$ must form a one-dimensional grid.
We use a standard reduction that allows us to assume that $G$ has maximum degree $3$, and then we argue that there exists some cut $(S',V\setminus S')$ of $G$ with $|\delta_G(S')|=2$, such that we can perform the above recursive procedure with respect to either $S'$ or $V\setminus S'$.

\section{Preliminaries and Notation}
All of the graphs that we consider in this work are undirected.
Given a graph $G=(V,E)$, a function $c: E \to \R$, and an edge $\{v,w\} \in E$, we overload notation and use $c(v,w)$ to denote $c(\{v,w\})$.
For any edge set $F\subseteq E$, we use $c(F)$ to denote $\sum_{e\in F}c(e)$.
Given a vertex set $S\subseteq V$, we use $\delta_G(S)$ to denote the set of edges of $G$ with precisely one endpoint in $S$.
We sometimes use the term \emph{path} interchangeably to mean either its edge set or its vertex sequence; the intended meaning will be clear from the context.
For any universe $U$, any $S\subseteq U$, and any $e\in U$, we sometimes use $S+e$ and $S-e$ to denote $S\cup\{e\}$ and $S\setminus\{e\}$, respectively.
For any integer $n\in\Zp$, we use $[n]$ to denote $\{1,2,\hdots,n\}$.

\subsection{Multicommodity Flows}
An instance $(G,u,H,d)$ of {\em multicommodity flow}
is given by an undirected graph $G=(V,E(G))$ with edge capacities $u:E(G)\to\mathbb{R}_{\geq 0}$,
and a graph $H=(V,E(H))$ with demand values $d:E(H)\to \mathbb{R}_{\geq 0}$. 
$G$ and $H$ may have parallel edges but no loops.
We refer to $G$ and $E(G)$ as the {\em supply graph} and {\em supply edges}, respectively. 
We refer to $H$ and $E(H)$ as the {\em demand graph} and {\em demand edges}, respectively. 

For a demand edge $h=\{s,t\} \in E(H)$, we use $\mathcal{P}_{h}$ to denote the set of all (simple) $st$-paths in $G$.
A flow for $h$, or {\em $h$-flow}, is an assignment $x_h:\mathcal{P}_h \to \Rp$ of non-negative real numbers to paths in $\mathcal{P}_h$ such that $\sum_{P\in\mathcal{P}_h}x_h(P)=d(h)$.
A collection of flows $x=(x_h)_{h\in E(H)}$ constitutes a {\em multicommodity flow} for the instance.
For simplicity, we often refer to multicommodity flows simply as flows.
For a supply edge $e \in E(G)$ and a demand edge $h \in E(H)$, we use $x_h(e):=\sum_{P: e \in P} x_h(P)$ to denote the total $h$-flow going through $e$.
We use $x(e):=\sum_{h \in E(H)} x_h(e)$ to denote the total flow going through $e$.
A flow $x$ is \emph{feasible} if it satisfies the edge-capacity constraints, that is, if $x(e) \leq u(e)$ for all $e \in E(G)$.

An $h$-flow $x_h$ is called \emph{unsplittable} if \emph{exactly one} of the paths in $\mathcal{P}_h$ is assigned a non-zero value in $x_h$.
We call a flow $x = (x_h)_{h \in E(H)}$ unsplittable if each $x_h$ is unsplittable.
In this work, whenever we refer to an unsplittable flow, we also assume that it routes all demands.
We use $\Dmax:=\max_{h\in E(H)} d(h)$ to denote the maximum demand value.
To simplify notation in the upcoming proofs,
we adopt the notion of $\alpha$-feasibility introduced in~\cite{aleman2025unsplittable} and generalize it to the notion of $\vec\beta$-feasibility.\footnote{
In~\cite{aleman26} the authors defined $(\alpha,\beta)$-feasibility, to indicate that a flow $x$ satisfies $x(e)\le \alpha\,u(e)+\beta\,\Dmax$ for every $e\in E(G)$ (and some $\alpha,\beta\in\Rp$).
In this work, we focus only on additive violations with respect to $\Dmax$.
} 

\begin{definition}[$\vec{\beta}$-feasibility]
For a vector $\vec{\beta}\in\Rp^{E(G)}$, we say that an (unsplittable) flow $x$ is {\em $\vec{\beta}$-feasible} if
\[x(e)\le u(e)+\vec\beta(e)\cdot \Dmax\text{ for each }e\in E(G);\]
we overload notation, and say that $x$ is {\em$\alpha$-feasible} for a real number $\alpha$, if $x(e) \leq u(e) + \alpha \cdot \Dmax$ for all $e \in E(G)$ (i.e., $x$ is $\vec{\beta}$-feasible, where $\vec{\beta}=\alpha\cdot\vec{1}$).
\end{definition}

We may assume w.l.o.g.\@ that the supply graph $G$ is $2$-vertex-connected.
Indeed, suppose that $G$ has a cut vertex.
Let $B_1 \subseteq V(G)$ denote the vertex set of a block\footnote{
A \emph{block} of $G$ is an inclusion-wise maximal induced subgraph of $G$ with no cut vertex.}
of $G$ containing exactly one cut vertex $v \in B_1$ and define $B_2 := V(G) \setminus (B_1 \setminus \{v\})$.
For any demand $h = \{a,b\} \in E(H)$ with $\{a,b\} \subseteq B_i$ for some $i \in \{1,2\}$ we have that any $ab$-path in $G$ is contained inside $B_i$, so that demand does not interact with edges in the other component $B_{3-i}$.
For a demand $h = \{a,b\} \in E(H)$ with $a \in B_1$, $b \in B_2$ we have that any $ab$-path in $G$ contains $v$.
Hence, we can replace $h$ by two demands $\{a,v\}$ and $\{v,b\}$, each of value $d(h)$, obtaining an equivalent multicommodity flow instance.
Repeating this transformation until no cut vertex exists reduces the problem to a collection of smaller instances, whose supply graphs are the blocks of $G$.

\subsection{Outerplanar Instances}\label{sec:outerplanar_instances}
A graph is outerplanar if it admits a planar embedding in which all vertices lie on the unbounded face.
Throughout this work, the supply graph $G$ of our instance $(G,u,H,d)$ is outerplanar, and we fix an outerplanar embedding of $G$.
We sometimes also call the unbounded face the \emph{outer face}, and the other faces \emph{inner faces}.
Similarly, we call the edges of $G$ incident with the outer face the \emph{outer edges}, and the remaining edges the \emph{inner edges}.
We denote the sets of outer and inner edges by $\Eouter(G)$ and $\Einner(G)$, respectively.
Since we may assume that $G$ is $2$-vertex-connected, every face of $G$ is bounded by a cycle.

\subsection{Cut-Condition}
Let $(G,u,H,d)$ denote a multicommodity flow instance,
and set $V := V(G)$.
For any $S\subseteq V$, the \emph{cut} $(S,V\setminus S)$ is a bipartition of the vertex set.
The \emph{cut-condition} requires that, for every cut, the total demand crossing the cut
is at most the total capacity of the supply edges crossing it i.e., \[u(\delta_G(S)) \ge d(\delta_H(S)) \quad \text{for all } S\subseteq V.\]
A set $\emptyset \neq S\subsetneq V$ is called \emph{central} if both of the induced graphs $G[S]$ and $G[V\setminus S]$ are connected.
We say that $(S,V\setminus S)$ is a \emph{central cut} if $S$ is central.
The following is a well-known fact (see, e.g.,~\cite{schrijver2003combinatorial}).
\begin{lemma}\label{lemma:cut_sufficiency}
An instance $(G,u,H,d)$ satisfies the cut-condition if and only if $u(\delta_G(S)) \ge d(\delta_H(S))$ for all central sets $S\subseteq V$.
\end{lemma}

Observe that in a $2$-vertex-connected outerplanar graph $G$,
the central sets are precisely the vertex sets of subpaths of the outer face.

A classical result of Okamura and Seymour states that the cut-condition is sufficient whenever
$G$ is planar and the endpoints of all demand edges lie on a common face of $G$.

\begin{theorem}[Okamura--Seymour~\cite{okamura1981multicommodity}]\label{theorem:OS}
Let $G$ be a planar supply graph and $H$ a demand graph.
Suppose there exists a face $F$ of $G$ such that for every demand edge $\{s,t\}\in E(H)$, 
both $s$ and $t$ lie on $F$.
Then, for any $u:E(G)\to\Rp$ and $d:E(H)\to\Rp$, the instance $(G,u,H,d)$ admits
a multicommodity flow if and only if it satisfies the cut-condition.
\end{theorem}
Since every vertex of an outerplanar graph is incident to the outer
face, an immediate consequence of the above theorem is that the cut-condition is sufficient for the feasibility of an instance if the supply graph $G$ is outerplanar.

\section{Improved Upper Bound for Outerplanar Instances}\label{section:2dmax}
In this section, we give a proof of Theorem~\ref{theorem:main}, which we restate for convenience.
\theoremmain*

Our main approach is to reduce our instance to two smaller instances, which can be solved recursively.
Ultimately, we arrive at a ring-loading problem, which is better understood than the case where $G$ is outerplanar.
We introduce the results on the ring-loading problem that we use in Section~\ref{sec:ring_loading}.
Afterwards, we prove Theorem~\ref{theorem:main}.

\subsection{Ring-Loading}\label{sec:ring_loading}

For the ring-loading problem, D\"aubel~\cite{daubel2019ringloadingBestUpperbound} showed how to compute a $1.3\,$-feasible unsplittable flow in any feasible instance.
For our purposes, however, it is essential to obtain a much stronger bound on the capacity violation for one special edge, while still keeping the capacity violations on all other edges small.
Therefore, we make use of the algorithm of Schrijver, Seymour, and Winkler~\cite{schrijver1998ringloading}, which only gives a $\frac{3}{2}$-feasible unsplittable flow, but also directly implies a bound of $+\frac{1}{2}\Dmax$ for the capacity violation on one edge that can be chosen arbitrarily in advance.

\begin{restatable}[Schrijver, Seymour, Winkler~\cite{schrijver1998ringloading}]{theorem}{theoremring}\label{theorem:ring_loading_3_2}
Let $(G, u, H, d)$ be a feasible ring-loading instance and $\etight \in E(G)$.
Then there is a $\vec\beta$-feasible unsplittable flow $y = (y_h)_{h \in E(H)}$,
where $\vec \beta(\etight) = \frac{1}{2}$ and $\vec \beta(e) = \frac{3}{2}$ for all other edges $e \in E(G) \setminus \{\etight\}$.
\end{restatable}
We also need an even stronger bound for the special case where all demand edges are incident to one of two specified vertices.

\begin{restatable}{lemma}{lemmaring}\label{lemma:ring_loading_with_two_demands}
    Let $\Iscr = (G, u, H, d)$ be a feasible ring-loading instance with two vertices $v,w \in V(G)$ such that any demand edge contains $v$ or $w$.
    Let $e_1, e_2 \in E(G)$.
    Then there exists a $\vec \beta$-feasible unsplittable flow $y = (y_h)_{h \in E(H)}$ for $\Iscr$, where $\vec \beta(e_1) = \vec \beta(e_2) = \frac{1}{2}$ and $\vec \beta(e) = \frac{3}{2}$ for $e \in E(G) \setminus \{e_1, e_2\}$.
\end{restatable}

The proof of Theorem~\ref{theorem:ring_loading_3_2} follows by using the same algorithm and analysis as Schrijver, Seymour, and Winkler~\cite{schrijver1998ringloading}.
We include the details, together with the proof of Lemma~\ref{lemma:ring_loading_with_two_demands}, in Appendix~\ref{appendix:omitted}.

\subsection{Proof of Theorem~\ref{theorem:main}}\label{sec:proof_main_thm}

Recall that a nonempty vertex set $S\subsetneq V(G)$ is central if $G[S]$ and $G[V(G) \setminus S]$ are connected.
A useful tool which we use in our proof is the notion of \emph{\nice{$x$} central sets}:

\begin{definition}[\nice{$x$} sets]\label{def:nice_cut}
    Let $(G, u, H, d)$ be an outerplanar multicommodity flow instance, and let $x = (x_h)_{h \in E(H)}$ be a feasible flow.
    We call a central vertex set $S \subseteq V(G)$ \emph{\nice{$x$}} if every $h \in E(H[S])$ is routed in $G[S]$, i.e., if $x_h(e) = 0$ for every edge $e \notin E(G[S])$.
\end{definition}

We make use of two types of \nice{$x$} sets, which are captured in the following two lemmata:

\begin{lemma}\label{lemma:tight_cuts_are_nice}
    Let $(G, u, H, d)$ be a multicommodity flow instance and $x = (x_h)_{h \in E(H)}$ a feasible flow.
    Let $(S, V(G) \setminus S)$ be a central cut that is tight, i.e., $u(\delta_G(S)) = d(\delta_H(S))$.
    Then $S$ is \nice{$x$}.
\end{lemma}
\begin{proof}
Since $x$ is feasible, we have
\[u(\delta_G(S))\geq x(\delta_G(S))= \sum_{h\in E(H)} x_h(\delta_G(S))
\geq \sum_{h\in \delta_H(S)} x_h(\delta_G(S))\geq d(\delta_H(S))= u(\delta_G(S)).\]
Therefore, equality holds throughout the above expression.
Thus, $x_h(\delta_G(S))=0$ for every $h\in E(H[S])$. Since both endpoints of such an $h$ lie in $S$, this implies that $x_h(e)=0$ for all $e\notin E(G[S])$. Therefore, $S$ is \nice{$x$}.
\end{proof}

\begin{lemma}\label{lemma:cuts_of_size_two_are_nice}
    Let $\Iscr = (G, u, H, d)$ be a multicommodity flow instance and $x = (x_h)_{h \in E(H)}$ a feasible flow that is minimal, i.e., for any feasible flow $x^\prime$ for $\Iscr$ with $x^\prime(e) \leq x(e)$ for all $e \in E(G)$ we actually have $x^\prime(e) = x(e)$ for all $e \in E(G)$.
    Let $(S_1, S_2 = V(G) \setminus S_1)$ be a central cut with $|\delta_G(S_1)| = 2$.
    Then $S_1$ or $S_2$ is \nice{$x$}.
\end{lemma}
\begin{proof}
    Set $C := \delta_G(S_1) = \delta_G(S_2)$. Assume that neither $S_1$ nor $S_2$ is \nice{$x$}.
    Then, for each $i=1,2$, there is a demand $h_i=\{s_i,t_i\}\in E(H[S_i])$ such that $x_{h_i}(C)\neq 0$.
    Choose an $s_i t_i$-path $P_i$ with $E(P_i)\cap C\neq\emptyset$ and $x_{h_i}(P_i)>0$.
    Since both endpoints of $h_i$ lie in $S_i$, the path $P_i$ contains both edges of $C$. Let $Q_i$ be the minimal subpath of $P_i$ containing both edges of $C$. 
    Then $Q_i-C$ is a path in $G[S_{3-i}]$. 
    For each $i=1,2$, let $P_i^\prime$ be the $s_i t_i$-path obtained from $P_i$ by replacing $Q_i$ with $Q_{3-i}-C$ and possibly deleting cycles.
    Now, decrease $x_{h_i}(P_i)$ by a sufficiently small amount $\varepsilon>0$ and increase $x_{h_i}(P_i^\prime)$ by $\varepsilon$, for $i=1,2$.
    This yields another feasible flow $x^\prime$ with $x^\prime(e)\leq x(e)$ for every $e\in E(G)$, and with $x^\prime(e)=x(e)-2\varepsilon$ for both edges $e\in C$, contradicting minimality of $x$.
\end{proof}

We can use an $x$-nice set in the following way to round the splittable flow $x$ to an unsplittable one:
First, we remove demands in $H[S]$, together with their flows, from our instance and replace $G[S]$ by a path (consisting of the outer edges of $G$ within $S$) with adequate edge capacities.
This yields a \emph{cut instance} (cf.\@ Definition~\ref{def:cut_instance}).
Due to the fact that $S$ is \nice{$x$},
the cut instance differs from the original instance only within $S$, so after solving the cut instance recursively, we can fix the obtained unsplittable flow $y^\prime$ everywhere except within $S$.
Afterwards, we need to route the demands in $H[S]$ and reconnect the segments of the unsplittable flows in $y^\prime$ that are routed in $S$.
All of this needs to be routed through $S$, so we encode this task in a \emph{split instance} (cf.\@ Definition~\ref{def:split_instance}), which we also solve recursively.

\begin{figure}
    \centering
    \includegraphics[width=0.4\textwidth]{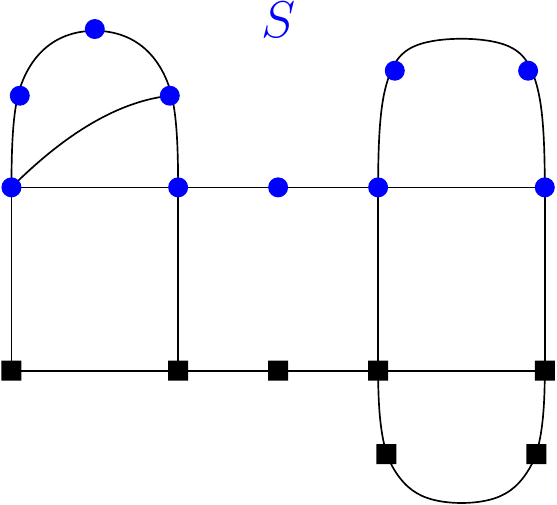}
    \hspace{10mm}
    {\includegraphics[width=0.4\textwidth]{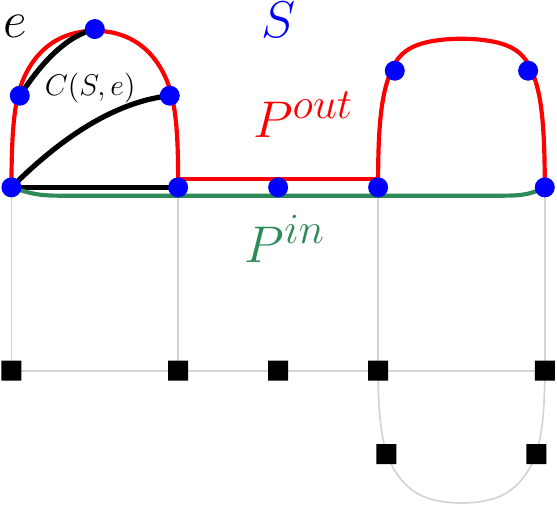}}

    \vspace{5mm}
    \caption{The left image shows an outerplanar graph with a central set $S$, whose vertices are blue and circular. The corresponding inner and outer paths are shown on the right (cf. Definition~\ref{def:outerinner}). Together, these paths form the outer edges of $G[S]$. For the indicated (outer) edge $e$, the semi-cut $C(S,e)$ consists of $e$, the leftmost edge of $\Pinner$, and the inner black edge of $G[S]$ between them (cf. Definition~\ref{def:capacity_definition_edges}).
    }\label{fig:inner_outer_path}
\end{figure}

\begin{definition}[Outer and inner paths]\label{def:outerinner}
    Let $G$ be a $2$-vertex-connected outerplanar graph and $S \subsetneq V(G)$ a central set.
    \begin{enumerate}[label=(\roman*)]
    \item The \emph{outer path} of $S$ is the unique path $\Pouter$ with $V(\Pouter)=S$ and $E(\Pouter)\subseteq \Eouter(G)$.
    \item For any $a, b \in S$ we define their outer path $\Pouter_{a,b}$ as the unique $ab$-subpath of $\Pouter$.
    \item The \emph{inner path} of $S$ is the path $\Pinner$ in $G[S]$ between the endpoints of $\Pouter$ whose edge set is the union of $\Eouter(G[S])\cap \Einner(G)$ and the bridges~\footnote{A bridge is an edge whose removal disconnects the graph.} of $G[S]$. 
    Equivalently, if $G[S]$ contains no parallel edges, $\Pinner$ is the unique shortest path in $G[S]$ between the endpoints of $\Pouter$.
    \end{enumerate}
    See Figure~\ref{fig:inner_outer_path} for an example.
\end{definition}

\begin{lemma}\label{remark:paths}
In the situation of Definition~\ref{def:outerinner} all of the following properties hold.
\begin{enumerate}[label=(\alph*)]
\item The outer path $\Pouter$ and inner path $\Pinner$ are well-defined.\label{itema}
\item The vertex sequence of $\Pinner$ is a subsequence of the vertex sequence of $\Pouter$.\label{itemb}
\item The edges in $\Pinner\cap\Pouter$ are precisely the bridges of $G[S]$.\label{itemc}
\item The edges in $\Pouter\cup\Pinner$ are precisely the outer edges of $G[S]$.\label{itemd}
\item The inner vertices of $\Pinner$ are precisely the cut-vertices of $G[S]$.\label{iteme}
\item If $|S|\geq2$, there is a one-to-one correspondence between the blocks of $G[S]$ and the edges of $\Pinner$. More precisely, each edge $\{v,w\}\in \Pinner$ corresponds to the block induced by the vertex set of $\Pouter_{v,w}$.\label{itemf}
\end{enumerate}
\end{lemma}
\begin{proof}
Assume w.l.o.g.\@ that the outer face of $G$ is bounded by the cycle $C=v_1,v_2,\hdots,v_n,v_1$, and that $v_1\in S$ and $v_n\notin S$.
Observe that there must be some $1\le \ell\le n-1$ such that $S=\{v_1,v_{2},\hdots,v_\ell\}$.
Indeed, otherwise there are some $a<b<c<d$ with $v_a,v_c\in S$ and $v_b,v_d\notin S$. Since $G[S]$ and $G[V\setminus S]$ are connected, there must be a $v_av_c$-path $P_{ac}$ in $G[S]$ and a $v_bv_d$-path $P_{bd}$ in $G[V\setminus S]$. Since the embedding of $P_{ac}$ and $P_{bd}$ lie inside the region bounded by $C$, some of the edges of these two paths would cross, contradicting the planarity of $G$.
This shows that $\Pouter$ is well-defined, and that its vertex sequence is given by $\Pouter=v_1,v_{2},\hdots,v_\ell$ for some $1\le\ell\le n-1$.

Note that $E(\Pouter)\subseteq \Eouter(G[S])$. We prove the lemma by induction on $|S|=\ell$.
It is easy to see that the lemma is satisfied if $\ell\in\{1,2\}$: if $\ell=1$, then $\Pouter$ and $\Pinner$ consist only of the vertex $v_1$, and if $\ell=2$, then $G[S]$ has a single block with vertex set $\{v_1,v_2\}$, possibly with parallel edges, so properties~\ref{itema}-\ref{itemf}  are immediate.

Let $j$ be the smallest index such that $\{v_j,v_\ell\}\in E(G)$. Note that $1\le j\le \ell-1$, since $v_{\ell-1}$ is adjacent to $v_\ell$. If $j=1$, then the outer face of $G[S]$ is bounded by the cycle $v_1,v_2,\ldots,v_\ell,v_1$.
Thus, $G[S]$ is 2-vertex-connected (and hence 2-edge-connected), and $\Pinner$ is the single edge $\{v_1,v_\ell\}$. Hence, all properties~\ref{itema}-\ref{itemf} are satisfied.
Assume now that $2\le j\le\ell-1$. Consider the central sets $S_1:=\{v_1,v_2,\hdots,v_j\}$ and $S_2:=\{v_{j},v_{j+1},\hdots,v_\ell\}$.
Let $\Pouter_1$ and $\Pinner_1$ denote the outer path and inner path of $S_1$, respectively.
By the induction hypothesis, $G[S_1]$ (and $S_1$, $\Pouter_1$, $\Pinner_1$) satisfy the lemma.

By the choice of $j$ and planarity, no edge of $G[S]$ has one endpoint in $S_1\setminus\{v_j\}$ and the other in $S_2\setminus\{v_j\}$.
Hence $G[S]$ is obtained by gluing $G[S_1]$ and $G[S_2]$ at the cut vertex $v_j$.
Moreover, $G[S_2]$ is a (leaf) block of $G[S]$. If $j=\ell-1$, this block has vertex set $\{v_{\ell-1},v_\ell\}$ and may contain parallel edges; otherwise, its outer face is bounded by the cycle $v_j,v_{j+1},\ldots,v_\ell,v_j$.
In either case, the inner path of $S_2$ consists of the edge $\{v_j,v_\ell\}$, and $\Pinner$ is obtained by concatenating $\Pinner_1$ with $\{v_j,v_\ell\}$.
Properties~\ref{itema}-\ref{itemf} for $G[S]$ now follow from the corresponding properties for $G[S_1]$, since the edge $\{v_j,v_\ell\}$ corresponds to the block $G[S_2]$, whose vertex set is precisely the vertex set of $\Pouter_{v_j,v_\ell}$.
\end{proof}

We now describe our procedure to compute an unsplittable routing by solving two smaller instances.
We first introduce the following useful notion, and then formally define the cut instance.

\begin{definition}[Semi-cut]\label{def:capacity_definition_edges}
Let $G$ be a $2$-vertex-connected outerplanar graph.
Let $S \subsetneq V(G)$ be a central set with $|S|\ge2$, and let $\Pouter$ be the outer path of $S$.
For any edge $e\in E(\Pouter)$ we define the \emph{semi-cut w.r.t.\@ $S$ and $e$} as the edge set $C(S,e) := \{\{v,w\} \in E(G[S]) : e \in E(\Pouter_{v,w})\}$.
Equivalently, $C(S,e) := \delta_{G[S]}(X)$, where $X \subseteq S$ is the vertex set of one of the two connected components of $\Pouter - e$.
See Figure~\ref{fig:inner_outer_path} for an example.
\end{definition}

\begin{definition}[Cut instance]\label{def:cut_instance}
Let $(G,u,H,d)$ be an outerplanar multicommodity flow instance where $G$ is
$2$-vertex-connected, and let $x=(x_h)_{h\in E(H)}$ be a feasible flow.
Let $S\subseteq V(G)$ be an \nice{$x$} central set with $|S|\ge2$, and let $\Pouter$ be the outer path of $S$.
The \emph{cut instance for $x$ and $S$} is the instance $\Icut=(G^\prime,u^\prime,H^\prime,d^\prime)$ defined as follows.
See Figure~\ref{fig:cut_instance} for an illustration.
The supply graph $G^\prime$ is given by
$V(G^\prime)=V(G)$ and $E(G^\prime) := \left(E(G)\setminus E(G[S])\right)\cup E(\Pouter)$.
The demand graph and demand values are given by
$E(H^\prime) := E(H)\setminus E(H[S])$ and $d^\prime(h):=d(h)$ for all $h\in E(H^\prime)$.
Edge capacities are defined as follows.
\[u^\prime(e):=\begin{cases}
 \sum_{e^\prime \in C(S,e)} \sum_{h \in E(H^\prime)} x_h(e^\prime)\,,&\:\text{if }~e\in\Pouter\\
u(e)\,,&\:\text{otherwise.}
\end{cases}\]   
\end{definition}

\begin{figure}
    \centering
    \includegraphics[width=0.4\textwidth]{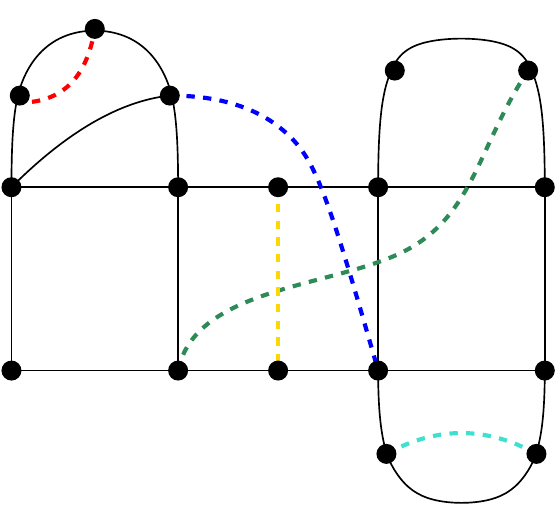}
    \hspace{8mm}
    \includegraphics[width=0.42\textwidth]{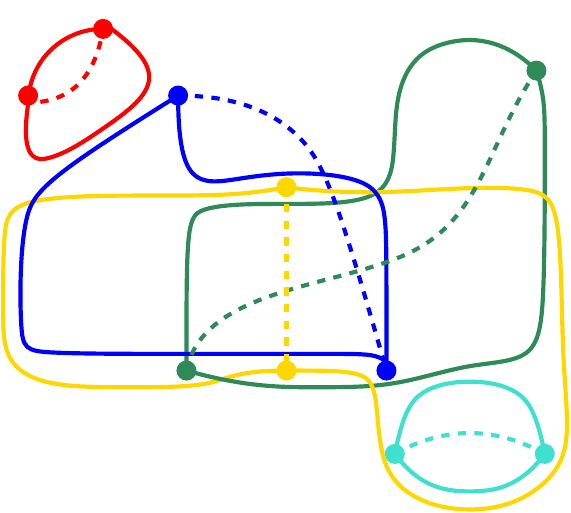}

    \vspace{5mm}
    \caption{The left image shows an outerplanar multicommodity flow instance (black edges) with 5 demands (colored dashed edges).
    The right image shows a feasible fractional flow, where the flow paths of each demand are drawn in the same color as the demand.
    }\label{fig:instance}
\end{figure}

\begin{figure}
    \centering
    \includegraphics[width=0.4\textwidth]{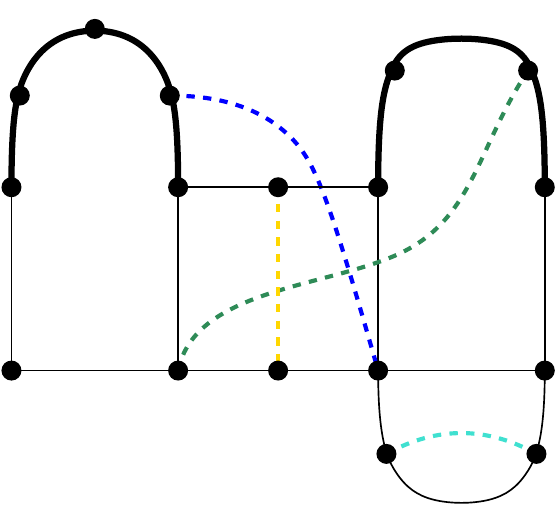}
    \hspace{8mm}
    \includegraphics[width=0.42\textwidth]{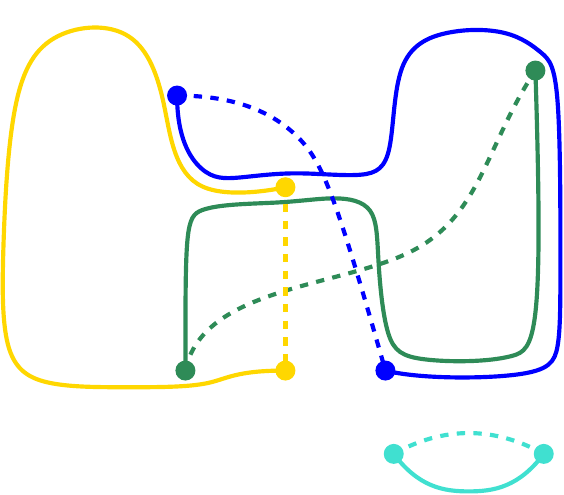}

    \vspace{5mm}
    \caption{The left image shows the cut instance for the instance in Figure~\ref{fig:instance}, where $S$ is chosen as in Figure~\ref{fig:inner_outer_path}.
    Note that the red demand is removed, and the capacities of inner edges $\{v,w\}$ in $S$ are added along their outer paths $\Pouter_{v,w}$.
    The right image shows a possible unsplittable flow for the cut instance.
    }\label{fig:cut_instance}
\end{figure}

\begin{figure}
    \centering
    \includegraphics[width=0.4\textwidth]{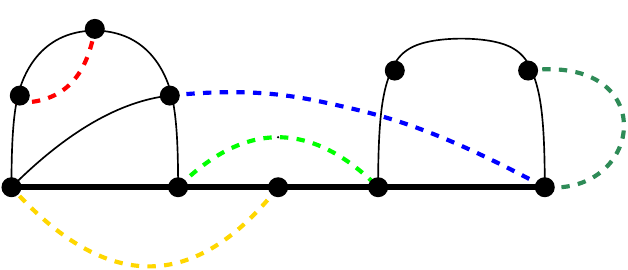}
    \hspace{8mm}
    \includegraphics[width=0.42\textwidth]{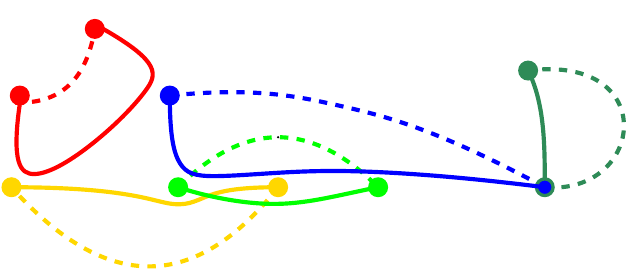}

    \vspace{5mm}
    \caption{The left image shows the split instance for the unsplittable flow from Figure~\ref{fig:cut_instance}.
    Note that the green demand from Figure~\ref{fig:cut_instance} induces two demands here.
    The capacity on the inner path (straight horizontal edges) is increased to make it feasible.
    An unsplittable flow in this instance can be combined with the unsplittable flow from Figure~\ref{fig:cut_instance} in order to obtain a solution for the original instance.
    }\label{fig:split_instance}
\end{figure}

\begin{figure}
    \centering
    \includegraphics[width=0.4\textwidth]{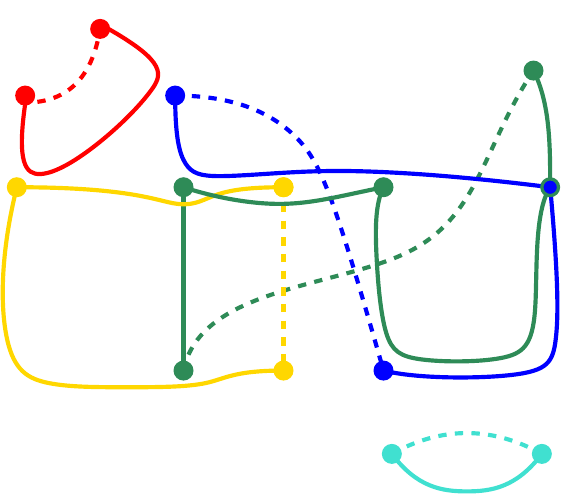}

    \vspace{5mm}
    \caption{Combining the unsplittable flows from Figure~\ref{fig:cut_instance} and Figure~\ref{fig:split_instance} yields an unsplittable flow for the original instance
    }\label{fig:unsplittable_solution}
\end{figure}

In other words, we first remove the flow of the demands that are routed completely inside $S$. We add the (residual) capacities of the inner edges of $G$ inside $S$ to the outer path of $S$, and then remove those inner edges from the instance.
This makes sure that we maintain feasibility of the cut instance:

\begin{lemma}\label{lemma:cut_instance_is_feasible}
    In the situation of Definition~\ref{def:cut_instance}, the cut instance $\Icut$ is feasible.
\end{lemma}
\begin{proof}
    We construct a feasible flow $x^\prime=(x^\prime_h)_{h\in E(H^\prime)}$ for $\Icut$ from the flow $x$.
    For any $h = \{s,t\} \in E(H^\prime)$ we construct $x^\prime_h$ as follows.
    For each $st$-path $P$ in $G$ that lies in the support of $x_h$,
    let $P^\prime$ be the $st$-path in $G^\prime$ obtained from $P$ by replacing each edge $\{v,w\} \in E(P)$ with $v,w \in S$ by their outer path $\Pouter_{v,w}$, and then deleting cycles if necessary.
    We then increase $x^\prime_h(P^\prime)$ by $x_h(P)$.
    Doing this for all $st$-paths in the support of $x_h$ yields a flow $x^\prime_h$ of value $d(h) = d^\prime(h)$ in $G^\prime$.
    It remains to check that the capacities of $\Icut$ are satisfied. If $e\in E(G^\prime)\setminus E(\Pouter)$, the above procedure does not increase the flow on $e$, and hence $x^\prime(e)\le x(e)\le u(e)=u^\prime(e)$.
    Now let $e\in E(\Pouter)$.
    By construction of $x^\prime$, for every $h\in E(H^\prime)$,
    \[x^\prime_h(e) \le \sum_{\substack{\{v,w\} \in E(G[S]) \\ e \in E(\Pouter_{v,w})}}x_h(\{v,w\}) = \sum_{\{v,w\}\in C(S,e)}x_h(\{v,w\}). \]
    Therefore,
    \[x^\prime(e)=\sum_{h\in E(H^\prime)}x_h^\prime(e)\le\sum_{h\in E(H^\prime)}\,\sum_{e^\prime\in C(S,e)}x_h(e^\prime)=u^\prime(e).\]
    Thus, $x^\prime$ is feasible for $\Icut$.
\end{proof}

Our rounding algorithm will find an adequate \nice{$x$} set such that the cut instance is simpler and can be solved recursively.
Consider an unsplittable flow for $\Icut$ specified by the paths $\{P_h\}_{h\in E(H')}$ of $G'$.
Next, we construct an unsplittable flow for $\Iscr$ as follows.
Each $h\in E(H[S])$ is to be routed inside of $G[S]$.
For each $h\in E(H')=E(H)\setminus E(H[S])$,
we fix the maximal subpaths of $P_h$ that lie outside $G[S]$,
and the remaining task is to connect the endpoints of our partial unsplittable $h$-flow inside $G[S]$.
This task is captured in the split instance, which we define now:

\begin{definition}[Split instance]\label{def:split_instance}
    Let $\Iscr=(G,u,H,d)$ be an outerplanar multicommodity flow instance where $G$ is $2$-vertex-connected, and let $x=(x_h)_{h\in E(H)}$ be a feasible flow. Let $S\subseteq V(G)$ be an \nice{$x$} central set with $|S|\ge2$, and let $\Pouter$ and $\Pinner$ be the outer path and inner path of $S$, respectively.
    Let $\Icut = (G^\prime, u^\prime, H^\prime, d^\prime)$ be the cut instance for $x$ and $S$.
    Let $y^\prime = (y^\prime_h)_{h \in E(H^\prime)}$ be a $\vec \beta$-feasible unsplittable flow of $\Icut$,
    for some $\vec \beta \in \Rp^{E(G^\prime)}$.
    
    The \emph{split instance for $\Iscr, \Icut$, $y^\prime$ and $\vec \beta$} is the instance $\Isplit = (\widehat G, \widehat u, \widehat H, \widehat d)$ defined as follows.
    First, we set $\widehat G := G[S]$.
    Edge capacities are given by
    \[\widehat u(v,w):=\begin{cases}
    u(v,w)+ \max\{\;\vec \beta(e^\prime) : e^\prime \in \Pouter_{v,w}\;\} \cdot \Dmax\,,&\:\text{if }~\{v,w\}\in E(\Pinner)\\
    u(v,w)\,,&\:\text{otherwise.}\\
    \end{cases}\]
    The demands $E(\widehat H)$ are constructed as follows.
    First, we add all the demands in $E(H[S])=E(H) \setminus E(H^\prime)$ to $E(\widehat H)$.
    Next, for each edge $h \in E(H^\prime)$, 
    let $P^\prime_h$ denote the path on which $y^\prime_h$ routes $h$.
    Let $\Qscr_h$ denote the set of maximal subpaths of $P^\prime_h$ that are in $G^\prime[S] = \Pouter$.
    For each $Q \in \Qscr_h$ with $E(Q) \neq \emptyset$ we add a demand $h_Q$ between the endpoints of $Q$ to $\widehat H$ and set $\widehat d(h_Q) := d(h)$.
    See Figure~\ref{fig:split_instance} for an example.
\end{definition}

A key property of the split instance is that it is always feasible, although we only increased capacities along the inner path $\Pinner$.

\begin{lemma}\label{lemma:split_instance_is_feasible}
    In the situation of Definition~\ref{def:split_instance} the split instance $\Isplit$ is feasible.
\end{lemma}
\begin{proof}
    We first observe that we may assume that $H^\prime=H$.
    If this is not the case,
    we could first remove all demands in $E(H) \setminus E(H^\prime)$ from $H$ and decrease all capacities $u(e)$ by $\sum_{h \in E(H) \setminus E(H^\prime)} x_h(e)$.
    This does not change the cut instance $\Icut$.
    Thus, if the lemma holds for the resulting split instance, then restoring the removed demands and their flows yields the lemma for the original split instance.

    Suppose that the outer face is bounded by the cycle $v_1,v_2,\hdots,v_n,v_1$.
    Assume w.l.o.g.\@ that $S=\{v_1,v_2,\hdots,v_k\}$ (and thus $\Pouter=v_1,v_2,\hdots,v_k$).
    Let $X \subseteq S$ denote the set of vertices that are adjacent to $V(G) \setminus S$ in $G$ (note that $X\subseteq V(\Pinner)$).
    By Lemma~\ref{lemma:cut_sufficiency} and Theorem~\ref{theorem:OS}, it suffices to show that the cut-condition holds for all central cuts of $\widehat G$.
    Let $U \subsetneq S$ be a central set of $\widehat G = G[S]$. Then $U = \{v_i, \dots, v_j\}$ for some $1 \leq i \leq j \leq k$.
    We distinguish two cases:
    \begin{description}
        \item[Case 1:] $i > 1$ and $j < k$.
        
        Since $S \setminus U$ is also a central set of $G[S]$, it follows that $U$ does not contain a cut vertex of $G[S]$; otherwise, $v_1$ would be disconnected from $v_k$ in $S\setminus U$.
        Since the cut vertices of $G[S]$ are precisely the inner vertices of $\Pinner$ (see Lemma~\ref{remark:paths}),
        and since $v_1,v_k\notin U$, it follows that $U\cap X\subseteq U\cap V(\Pinner)=\emptyset$.
        Therefore, $\delta_{\widehat G}(U) = \delta_G(U)$ and $u(\delta_G(U)) = \widehat u(\delta_{\widehat G}(U))$.        
        Note that each demand $h\in E(H')$ induces at most one demand $h_Q$ with an endpoint in $S\setminus X$. 
        If such an $h_Q$ exists, then this endpoint is the unique endpoint of $h$ inside $S$.
        Hence, since $U\cap X=\emptyset$, we have $\delta_{\widehat H}(U)=\delta_H(U)$.
        Thus, \[\widehat d(\delta_{\widehat H}(U)) = d(\delta_H(U)) \leq u(\delta_G(U)) = \widehat u(\delta_{\widehat G}(U)),\]
        where the inequality follows from the fact that the initial instance $\Iscr$ satisfies the cut-condition.
        \item[Case 2:] $i = 1$ or $j = k$.
        
        W.l.o.g.\@ $i=1$ and thus $j < k$.
        Define $e := \{v_j, v_{j+1}\}$, and observe that $\delta_{\widehat G}(U)$ is precisely the semi-cut $C(S,e)$ for $S$ and $e$ (see Definition~\ref{def:capacity_definition_edges}).
        Since exactly one endpoint of $\Pinner$ lies in $U$, there is some $e^\prime \in C(S,e) \cap E(\Pinner)$.
        We have $\widehat u(e^\prime) \geq u(e^\prime) + \vec \beta(e)\cdot \Dmax$.
        Thus, \[\widehat u(\delta_{\widehat G}(U)) \geq u(C(S,e)) + \vec \beta(e)\cdot \Dmax \geq u^\prime(e) + \vec \beta(e) \cdot \Dmax\geq y^\prime(e).\]
        Finally, observe that by definition of $\widehat H$ we have $\widehat d(\delta_{\widehat H}(U)) = y^\prime(e)$. \qedhere
    \end{description}
\end{proof}

We now have all the ingredients to prove Theorem~\ref{theorem:main}.
We prove a slightly stronger statement that enables us to use an inductive argument.
Rather than presenting the proof as an explicit induction, we use a minimal counterexample argument.

\begin{theorem}\label{thm:2dmax}
    Let $\Iscr = (G,u,H,d)$ be an outerplanar multicommodity flow instance that is feasible, or equivalently, that satisfies the cut-condition.
    Let $\Etight \subseteq \Eouter(G)$ be a set of outer edges of $G$ that contains at most one edge of each block of $G$.
    Then there exists a $\vec \beta$-feasible unsplittable flow $y = (y_h)_{h \in E(H)}$, where
    \[
    \vec \beta(e) :=
    \begin{cases}
    \frac{1}{2}, & \text{if } e \in \Etight,\\
    \frac{3}{2}, & \text{if } e \in \Eouter(G) \setminus \Etight,\\
    2, & \text{if } e \in \Einner(G).
    \end{cases}
    \]
\end{theorem}
\begin{proof}
    Assume that this statement is false, and let $(G, u, H, d)$, together with $\Etight \subseteq E(G)$, be a counterexample that lexicographically minimizes $(|\Einner(G)|, |E(H)|)$.
    As discussed in Section~\ref{sec:outerplanar_instances}, the blocks of $G$ are independent, so we may assume that $G$ is $2$-vertex-connected.
    In particular, $|\Etight| \leq 1$ and we can assume w.l.o.g.\@ that $\Etight = \{\etight\}$ for some outer edge $\etight \in \Eouter(G)$.

    Furthermore, we may assume that for every edge $e \in E(G)$ there exists a central cut $(S_e, V(G) \setminus S_e)$ with $e \in \delta_G(S_e)$ that is tight, i.e.,
    $u(\delta_G(S_e)) = d(\delta_H(S_e)).$
    Otherwise, we can decrease $u(e)$ while preserving the cut-condition, until either such a cut exists or $u(e)=0$.
    The latter case can be ruled out by the minimality of the counterexample.
    Let $S := S_{\etight}$ and let $x = (x_h)_{h \in E(H)}$ be a feasible flow for $\Iscr$.
    By Lemma~\ref{lemma:tight_cuts_are_nice}, both $S$ and $V(G) \setminus S$ are \nice{$x$}.
    Also, the fact that each edge belongs to a tight cut implies that $x$ is minimal, i.e., if $x^\prime$ is feasible with $x^\prime(e) \leq x(e)$ for all $e \in E(G)$ then also $x^\prime(e) = x(e)$ for all $e \in E(G)$.

    \begin{description}
    \item[Case 1:] $G[S]$ or $G[V(G) \setminus S]$ contains an inner edge of $G$.
    
     W.l.o.g.\@ $E(G[S]) \cap \Einner(G) \neq \emptyset$.
     Let $\Pouter$ and $\Pinner$ denote the outer path of $S$ and inner path of $S$, respectively.
     Let $\Icut = (G^\prime, u^\prime, H^\prime, d^\prime)$ be the cut instance for $x$ and $S$.
     
     Observe that $|\Einner(G^\prime)| < |\Einner(G)|$.
     Thus, by minimality of our counterexample,
     we can assume that there is a $\vec \beta^\prime$-feasible unsplittable flow $y^\prime = (y^\prime_h)_{h \in E(H^\prime)}$ for $\Icut$,
     where $\vec \beta^\prime(\etight) = \frac{1}{2}$, $\vec \beta^\prime(e) = \frac{3}{2}$ for $e \in \Eouter(G^\prime) \setminus \{\etight\}$ and $\vec \beta^\prime(e) = 2$ for $e \in \Einner(G^\prime)$.
    
     Let $\Isplit = (\widehat G, \widehat u, \widehat H, \widehat d)$ be the split instance for $\Iscr$, $\Icut$, $y^\prime$ and $\vec \beta^\prime$ according to Definition~\ref{def:split_instance}.
     By Lemma~\ref{lemma:split_instance_is_feasible}, $\Isplit$ is feasible.
     Recall that $\Eouter(\widehat G)=\Pouter\cup\Pinner$, and that for any edge $\{v,w\}\in \Pinner$, the vertex set $V(\Pouter_{v,w})$ is a block of $\widehat G$.
     Therefore, $\Pinner$ contains precisely one edge from every block of $\widehat G$.
     Since $|\Einner(\widehat G)| < |\Einner(G)|$, by minimality of our counterexample we can find a $\widehat {\vec \beta}$-feasible unsplittable flow $\widehat y = (\widehat y_h)_{h \in E(\widehat H)}$ for $\Isplit$,
     where $\widehat {\vec \beta}(e) = \frac{1}{2}$ for $e \in E(\Pinner)$, $\widehat {\vec \beta}(e) = \frac{3}{2}$ for $e \in \Eouter(\widehat G) \setminus E(\Pinner)$ and $\widehat {\vec \beta}(e) = 2$ for $e \in \Einner(\widehat G)$.

     Now, we can define the unsplittable flow $y=(y_h)_{h\in E(H)}$ for the original instance $(G,u,H,d)$:
     For $h \in E(H) \setminus E(H^\prime)=E(H[S])$ we also have $h \in E(\widehat H)$, and hence $y_h := \widehat y_h$ defines an unsplittable $h$-flow.
     For $h \in E(H^\prime)$ let $P^\prime_h$ denote the path on which $y^\prime_h$ routes $h$.
     As in Definition~\ref{def:split_instance},
     let $\Qscr_h$ define the set of maximal subpaths of $P^\prime_h$ that are in $G^\prime[S] = \Pouter$.
     For each $Q \in \Qscr_h$ we have added a demand $h_Q$ to $E(\widehat H)$; we construct the path $P_h$ from $P^\prime_h$ by replacing each $Q \in \Qscr_h$ by the path that $h_Q$ is routed on in the unsplittable flow $\widehat y$, and deleting possible cycles in the constructed walk.
     We define $y_h$ to be the unsplittable $h$-flow that routes $h$ along $P_h$ (cf.\@ Figure~\ref{fig:unsplittable_solution}).
         
     It is left to check the capacity constraints for $y$.
     For $e \in E(G) \setminus E(G[S])$ we have 
     \[ y(e) \leq y^\prime(e) \leq u^\prime(e) + \vec \beta^\prime(e) \cdot \Dmax = u(e) + \vec \beta(e) \cdot \Dmax.\]
     Now consider an edge $e \in E(G[S])$.
     We have 
     \[y(e) \leq \widehat y(e) \leq \widehat u(e) + \widehat {\vec \beta}(e) \cdot \Dmax.\]
     If $e \notin E(\Pinner)$ then $\widehat u(e) = u(e)$ and $\widehat {\vec \beta}(e) = \vec \beta(e)$, so $y(e) \leq u(e) + \vec \beta(e) \cdot \Dmax$.
     Finally, consider the case $e \in E(\Pinner)$.
     By our choice of $\vec \beta^\prime$, Definition~\ref{def:split_instance} implies $\widehat u(e) \leq u(e) + \frac{3}{2} \cdot \Dmax$ in this case.
     If $e \in E(\Pinner)$ is an outer edge of $G$ then $e$ must be a bridge in $\widehat G$;
     in particular we even have $\widehat y(e) \leq \widehat u(e) \leq u(e) + \frac{3}{2} \cdot \Dmax = u(e) + \vec \beta(e)\cdot\Dmax$.
     Otherwise, we have $\widehat {\vec \beta}(e) = \frac{1}{2} = \vec \beta(e) - \frac{3}{2}$.
     Therefore,
     \[y(e)\le \widehat u(e)+ \widehat {\vec \beta}(e) \cdot \Dmax \le u(e)+\tfrac{3}{2}\cdot\Dmax + \left(\vec\beta(e)-\tfrac{3}{2}\right)\cdot\Dmax= u(e)+\vec\beta(e)\cdot\Dmax.\]
     This concludes the proof for case 1.
    \item[Case 2:] $\delta_G(S)$ contains all inner edges of $G$.\footnote{In this case, deleting the vertex corresponding to the outer face from the planar dual results in a path.}
    
     Let $F$ be the (unique) inner face of $G$ that is incident to $\etight$.
     If $F$ is the only inner face of $G$ then $G$ is a cycle, and Theorem~\ref{theorem:ring_loading_3_2} yields a $\vec \beta$-feasible unsplittable flow.
     Otherwise, since $\delta_G(S)$ contains $\etight$ and all inner edges of $G$, $F$ is incident to exactly one inner edge, say $\{v,w\} \in \Einner(G)$.
     Let $v_1, v_2 \in V(G)$ such that $\{v_1, v\}$ and $\{v_2, v\}$ are the two outer edges of $G$ incident to $v$, where $v_1$ lies on the boundary of $F$. Let $w_1, w_2 \in V(G)$ be defined analogously for $w$.
     
     We can assume w.l.o.g.\@ that $|\delta_G(v)| = 3$ and that no demand edges of $H$ are incident to $v$;
     otherwise, replace $v$ by two vertices $v$ and $v^\prime$ that are connected by an edge $\{v,v^\prime\}$ and connect all edges in $\delta_{G+H}(v)$ except $\{v_1, v\}$ and $\{v,w\}$ to $v^\prime$ instead of $v$.
     We can choose the capacity for the new edge $\{v,v^\prime\}$ large enough such that our instance is feasible.
     Note that any $\vec\beta$-feasible unsplittable flow in the constructed instance directly induces a $\vec \beta$-feasible unsplittable flow in the original instance.
     Similarly, we can assume that $|\delta_G(w)| = 3$ and that no demand edges of $H$ are incident to $w$.
     Now, we again distinguish two cases:
     \begin{description}
         \item[Case 2a:] There exists a demand edge $h$ between two vertices on the boundary of $F$.
         
         Observe that $x_h(\{v_1, v\}) = x_h(\{w_1,w\})$.
         Note that if $x_h(\{v_1,v\}) = 0$ then $x$ routes all the demand of $h$ along a single path (on the boundary of $F$).
         Thus, removing $h$ from $H$ as well as removing $x_h$ from $x$ and $u$ yields a counterexample with fewer demand edges, contradicting minimality of our counterexample.
         So w.l.o.g.\@ assume $x_h(\{v_1,v\}) > 0$.
         By Lemma~\ref{lemma:cuts_of_size_two_are_nice} one connected component $S \subseteq V(G)$ of $G- \{\{v_1, v\}, \{w_1,w\}\}$ is \nice{$x$}, and due to the flow $x_h$ we know that $v, w \in S$.
         In particular, the supply graph of the cut instance for $x$ and $S$ has less inner edges than $G$, so we can proceed exactly as in Case 1 to finish the proof in this case.
         \item[Case 2b:] Any demand edge contains at most one vertex on the boundary of $F$.
         
         In this case we consider the set $S$ of all vertices on the boundary of $F$.
         $S$ is \nice{$x$} because $H[S] = \emptyset$.        
         Let $\Icut = (G^\prime, u^\prime, H^\prime, d^\prime)$ be the cut instance for $x$ and $S$.
         We have $E(G^\prime) = E(G) \setminus \{\{v,w\}\}$, so by minimality of our counterexample there exists a $\vec \beta^\prime$-feasible unsplittable flow $y^\prime = (y^\prime_h)_{h \in E(H^\prime)}$ for $\Icut$, where $\vec \beta^\prime(e) = \frac{3}{2}$ for $e \in \Eouter(G^\prime)$ and $\vec \beta^\prime(e) = 2$ for $e \in \Einner(G^\prime)$.
         Now let $\Isplit = (\widehat G, \widehat u, \widehat H, \widehat d)$ be the split instance for $\Iscr$, $\Icut$, $y^\prime$ and $\vec \beta^\prime$.
         Clearly, $\widehat G$ is a cycle corresponding to the boundary edges of $F$.
         Furthermore, the assumption of case 2b implies that each edge of $\widehat H$ is incident to $v$ or $w$.
         Therefore, we can apply Lemma~\ref{lemma:ring_loading_with_two_demands} to find a $\widehat {\vec \beta}$-feasible unsplittable flow $\widehat y$ for $\Isplit$, where $\widehat {\vec \beta}(\etight) = \widehat {\vec \beta}(\{v,w\}) = \frac{1}{2}$ and $\widehat {\vec \beta}(e) = \frac{3}{2}$ for all other $e \in E(\widehat G)$.
         
         Define the unsplittable flow $y$ as in case 1.
         Analogously to case 1, it is straightforward to verify that the flow $y$ fulfills the requirements of the lemma. \qedhere
     \end{description}
    \end{description}
\end{proof}

Since $\vec \beta(e) \leq 2$ for all edges $e \in E(G)$ in Theorem~\ref{thm:2dmax}, this directly implies the existence part of Theorem~\ref{theorem:main}.
Note that although the proof of Theorem~\ref{thm:2dmax} is existential, it still shows a clear way to obtain an unsplittable flow $y$ as desired:
In each iteration, after partitioning the instance into the blocks of $G$ and removing demands which are already routed unsplittably, we construct two smaller instances, the cut instance $\Icut$ and the split instance $\Isplit$, and solve them recursively.
Afterwards, we combine the obtained unsplittable flows to the flow $y$.
Note that the combined number of inner faces of the supply graphs in $\Icut$ and $\Isplit$ equals the number of inner faces of $G$, so the total number of recursion steps is bounded linearly.
Thus, we can also compute an unsplittable flow as guaranteed by Theorem~\ref{thm:2dmax} in polynomial time.

\section{Improved Lower Bound for Outerplanar Instances}
We prove Theorem~\ref{theorem:lower} in this section, which we restate for convenience.

\theoremlower*

This section is organized as follows.
As a warm-up, we show in subsection~\ref{section:warmup} that the additive violation of edge capacities has to be at least $\tfrac{5}{4}\Dmax$ for outerplanar graphs.
This weaker lower bound already improves over the previous best lower bound of $\tfrac{11}{10}\Dmax$ in~\cite{skutella2016ringloading} for outerplanar graphs, 
which is achieved in a ring-loading instance.
Building on the ideas in subsection~\ref{section:warmup}, we prove Theorem~\ref{theorem:lower} in subsection~\ref{section:theoremlower},
whose construction is significantly more intricate.

\subsection{Warm Up: Lower Bound of $+\frac{5}{4}\Dmax$}\label{section:warmup}

In this subsection, we construct an instance on an outerplanar graph $G$ with five inner faces. Four of these faces are attached to a central face. We show that any unsplittable flow on this instance exceeds the capacity of some edge by at least $\tfrac{5}{4}\,\Dmax$.
The construction in the proof of Theorem~\ref{theorem:lower} generalizes this, by considering an instance in which $G$ consists of four 1-dimensional grids attached to a central face.

\begin{figure}[h]
\centering
\includegraphics[width=0.3\textwidth]{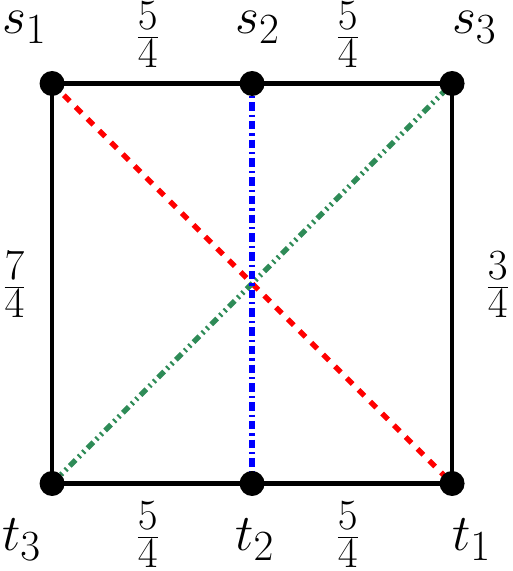}

\vspace{7mm}

\caption{The ring-loading instance $\mathcal{R}$}\label{fig:ring_loading_instance_r}
\end{figure}

We first consider the following ring-loading instance $\mathcal{R}:=(C,u,H,d)$, illustrated in Figure~\ref{fig:ring_loading_instance_r}.
\begin{align*}
\textbf{Supply graph:}&\quad C=s_1,s_2,s_3,t_1,t_2,t_3,s_1\text{ is a 6-cycle}\\[1ex]
\textbf{Demand graph:}&\quad E(H)=\{\{s_1,t_1\},\,\{s_2,t_2\},\,\{s_3,t_3\}\}\\[1ex]
\textbf{Capacities:}&\quad u(s_1,s_2)=\tfrac{5}{4},\;u(s_2,s_3)=\tfrac{5}{4},\;u(s_3,t_1)=\tfrac{3}{4}\\[0.5ex]
&\quad u(t_1,t_2)=\tfrac{5}{4},\;u(t_2,t_3)=\tfrac{5}{4},\;u(t_3,s_1)=\tfrac{7}{4}\\[1ex]
\textbf{Demands:}&\quad d(s_1,t_1)=1,\quad d(s_2,t_2)=\tfrac{1}{2},\quad d(s_3,t_3)=1.
\end{align*}
\begin{lemma}\label{lemma:feasiblering}
$\mathcal{R}$ is a feasible instance.
\end{lemma}
\begin{proof}
Consider the following feasible flow.
For each $i\in\{1,2,3\}$, let $r_i$ denote the amount of $s_i$-$t_i$ flow routed (clockwise) along the path $s_i,\hdots,s_3,t_1,\hdots,t_i$, and let $\overline{r_i}$ denote the amount of $s_i$-$t_i$ flow routed (counter-clockwise) along the path $s_i,s_{i-1},\hdots,s_1,t_3,\hdots,t_i$.
Take
\[r_{1}=r_{3}=\frac{1}{4},\qquad \overline{r_{1}}=\overline{r_{3}}=\frac{3}{4},\qquad r_{2}=\overline{r_{2}}=\frac{1}{4}.\]
Observe that the total amount of flow on each edge of the paths $s_1,s_2,s_3$ and $t_1,t_2,t_3$ is equal to $\tfrac{5}{4}$. The flow on $\{s_3,t_1\}$ is equal to $\tfrac{3}{4}$, and the flow on $\{t_3,s_1\}$ is equal to $\tfrac{7}{4}$.
\end{proof}

The following is a simple observation that we later generalize to the classes of instances considered in the next section (cf. Lemma~\ref{lemma:totheleft}).

\begin{lemma}\label{claim:warmup}
Let $y$ be an unsplittable flow of $\mathcal{R}=(C,u,H,d)$.
If for some some $i\in\{1,3\}$,
$y$ routes $\{s_i,t_i\}$ (clockwise) along the path $s_i,\hdots,s_3,t_1,\hdots,t_i$, then there is an edge $e\in E(G) \setminus \{\{t_3,s_1\}\}$ with 
\[y(e)\ge u(e)+\tfrac{5}{4}.\]
\end{lemma}
\begin{proof}
Let $P_i$ denote the path on which $y$ routes $\{s_i,t_i\}$ for $i\in\{1,2,3\}$.
Observe that if both $\{s_1,t_1\}$ and $\{s_3,t_3\}$ are routed clockwise, then $\{s_3,t_1\}\in P_1\cap P_3$. Thus,
\[y(s_3,t_1)\ge 2=u(s_3,t_1)+\tfrac{5}{4}.\]
Suppose now that $\{s_1,t_1\}$ is routed clockwise and $\{s_3,t_3\}$ is routed counter-clockwise. Then $\{s_1,s_2\},\{s_2,s_3\}\in P_1\cap P_3$. If $\{s_2,t_2\}$ is routed clockwise, then all three demands are routed through $\{s_2,s_3\}$ and hence
\[y(s_2,s_3)=\tfrac{5}{2}=u(s_2,s_3)+\tfrac{5}{4}.\]
Similarly, if $\{s_2,t_2\}$ is routed counter-clockwise, then 
\[y(s_1,s_2)=\tfrac{5}{2}=u(s_1,s_2)+\tfrac{5}{4}.\]
The case in which $\{s_1,t_1\}$ is routed counter-clockwise and $\{s_3,t_3\}$ is routed clockwise is analogous. In that case, the capacity of $\{t_1,t_2\}$ or $\{t_2,t_3\}$ is exceeded by at least $\tfrac{5}{4}$.
\end{proof}

By the above lemma, in order to avoid congesting an edge with (at least) $\tfrac{5}{4}$ units of flow, both $\{s_1,t_1\}$ and $\{s_3,t_3\}$ must be routed counter-clockwise.
Observe that in this case the capacity of $\{t_3,s_1\}$ is exceeded by at least $\tfrac{1}{4}$. We use this observation in the construction that we consider next.

Consider four vertex-disjoint copies $\mathcal{R}^{(\ell)}=\{(C^{\ell},u^{\ell},H^{\ell},d^{\ell})\}_{\ell\in\{a,b,c,r\}}$ of $\mathcal{R}$.
For each $\ell\in\{a,b,c,r\}$,
we use $$C^{\ell}=s_1^\ell,s_2^\ell,s_3^\ell,t_1^\ell,t_2^\ell,t_3^\ell,s_1^\ell$$ to denote the supply cycle of $\mathcal{R}^{(\ell)}$ (and we interpret $s_i^\ell\equiv s_i$ and $t_{i}^\ell\equiv t_i$).
We will construct a bigger instance $\mathcal{I}$ in which we add one unit of capacity to the edge $\{t_3^\ell,s_1^\ell\}$ of each of these copies,
and then attach it to a ``central'' cycle, whose remaining edges have unit capacities.
More concretely,
create four new vertices $w_1,w_2,w_3,w_4$,
and consider the $12$-cycle $$W=w_1,t^{a}_3,s^a_1,w_2,t^{b}_3,s^b_1,w_3,t^{c}_3,s^c_1,w_4, t^{r}_3,s^r_1,w_1$$
(whose nodes can be viewed as being ordered counter-clockwise around $W$).
We add two ``crossing'' unit demand edges $\{w_1,w_3\}$ and $\{w_2,w_4\}$ of demand $1$,
and then use Lemma~\ref{claim:warmup} to argue that any unsplittable flow will route at least $u(e)+\tfrac{5}{4}$ units of flow on some edge $e$.

Formally, consider the following instance $\mathcal{I}=(G,u,H,d)$. See Figure~\ref{fig:lower_bound_5_4}.
\begin{align*}
\textbf{Vertex set:}&\quad V=\{w_1,w_2,w_3,w_4\}\cup \bigcup_{\ell\in\{a,b,c,r\}} V(C^{\ell})\;\;.\\
\textbf{Supply edges:}&\quad E(G)= E(W)\cup \bigcup_{\ell\in\{a,b,c,r\}} \big(\;E(C^{\ell}) \setminus \{\{t_3^\ell,s_1^\ell\}\}\;\big)\;\;.\\[1ex]
\textbf{Demand edges:}&\quad E(H)=\{\{w_1,w_3\},\;\{w_2,w_4\}\,\} \cup \{\{s^\ell_1,t^\ell_1\},\,\{s^\ell_2,t^\ell_2\},\,\{s^\ell_3,t^\ell_3\}\}_{\ell\in\{a,b,c,r\}} \;\;.\\[1ex]
\textbf{Capacities:}
&\quad u(s^\ell_1,s^\ell_2)=u(s^\ell_2,s^\ell_3)=\tfrac{5}{4}\,,\quad u(s^\ell_3,t^\ell_1)=\tfrac{3}{4}\,,\quad
u(t^\ell_1,t^\ell_2)=u(t^\ell_2,t^\ell_3)=\tfrac{5}{4}\,,\\[1ex]
&\quad u(t^\ell_3,s_1^\ell)=1+u^\ell(t^\ell_3,s_1^\ell)=\tfrac{11}{4}\;\qquad \qquad\forall\ell\in\{a,b,c,r\}\;.\\[1ex]
&\quad u(e)=1,\qquad \forall e\in E(W)\setminus\{\{t_3^\ell,s_1^\ell\}\;:\;\ell=a,b,c,r\}.\\[1ex]
\textbf{Demands:}
&\quad d(s^\ell_1,t^\ell_1)=1,\quad d(s^\ell_2,t^\ell_2)=\tfrac{1}{2},\quad d(s^\ell_3,t^\ell_3)=1,\quad \forall \ell\in\{a,b,c,r\}.\\[1ex]
&\quad d(w_1,w_3)=1,\quad d(w_2,w_4)=1\quad.
\end{align*}
First observe that $\mathcal{I}$ is feasible.
\begin{lemma}
$\mathcal{I}=(G,u,H,d)$ is a feasible instance.
\end{lemma}
\begin{proof}
For each $\ell\in\{a,b,c,r\}$, route each demand in $\{\{s_1^\ell,t_1^\ell\},\,\{s_2^\ell,t_2^\ell\},\,\{s_3^\ell,t_3^\ell\}\}$ on the edge set of $C^\ell$ precisely as in the proof of Lemma~\ref{lemma:feasiblering}.
Since the edge capacities $u^\ell,\,u$ assigned to $E(C^\ell)$ in $\mathcal{R}^{(\ell)}$ and $\mathcal{I}$, respectively, differ only in the fact that $u(t_3^\ell,s_1^\ell)=u^\ell(t_3^\ell,s_1^\ell)+1$, it follows that after routing these twelve demands no edge capacity is exceeded in $\mathcal{I}$,
and that every (supply) edge of $E(W)$ still has one unit of capacity left.
Thus, we can route $\{w_1,w_3\}$ (resp. $\{w_2,w_4\}$) by sending $\tfrac{1}{2}$ units of flow along each of the two $w_1w_3$-paths (resp. $w_2w_4$-paths) in $W$. 
\end{proof}

\begin{figure}[h]
    \centering
    \includegraphics[width=0.55\textwidth]{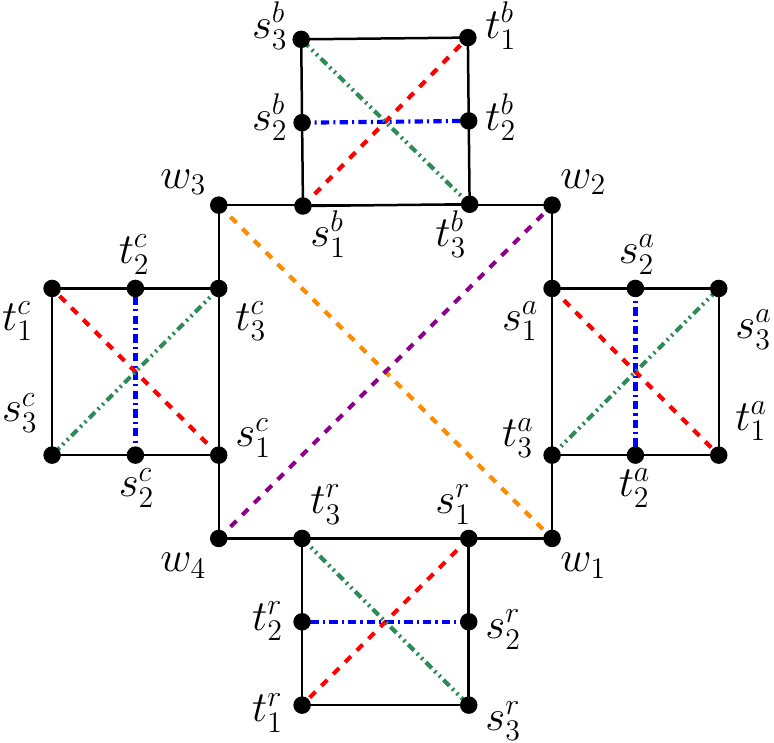}

    \vspace{7mm}
    \caption{The instance $\Iscr$ defined below}\label{fig:lower_bound_5_4}
\end{figure}

We are ready to prove the following weaker version of Theorem~\ref{theorem:lower}.
\begin{theorem}\label{theorem:weaker}
There exists a feasible instance $(G,u,H,d)$ on an outerplanar supply graph $G$, such that for any unsplittable flow $y$, there exists an edge $e\in E(G)$ with
\[y(e)\ge u(e)+\tfrac{5}{4}\,\Dmax.\]
\end{theorem}
\begin{proof}
Consider the instance $\mathcal{I}$ defined above. Note that $\Dmax=1$.
Let $P_{1,3}$ and $P_{2,4}$ denote the paths on which $y$ routes $\{w_1,w_3\}$ and $\{w_2,w_4\}$, respectively.
Then, there exist some $i\in[4]$ and $\ell\in\{a,b,c,r\}$ such that $\{w_i,t_3^\ell\},\{s^\ell_1,w_{i+1}\}\in P_{1,3}\cap P_{2,4}$ (where we interpret $w_5\equiv w_1$).
By the symmetry of the instance, we may assume without loss of generality that $\{w_1,t^a_3\},\{s^a_1,w_2\}\in P_{1,3}\cap P_{2,4}$.
We show that the capacity of one of the edges in the copy $C^a$ is exceeded by at least $\tfrac{5}{4}$.
The proof for the other three cases is identical.
Since  $\{w_1,t^a_3\},\{s^a_1,w_2\}\in P_{1,3}\cap P_{2,4}$, then 
\begin{align*}
y(w_1,t_3^a)&\ge y_{\{w_1,w_3\}}(w_1,t_3^a)+y_{\{w_2,w_4\}}(w_1,t_3^a)=2,\\
y(s_1^a,w_2)&\ge y_{\{w_1,w_3\}}(s_1^a,w_2)+y_{\{w_2,w_4\}}(s_1^a,w_2)=2.
\end{align*}
Observe that if $y$ routes any demand $h\in\{\{s_1^a,t_1^a\},\,\{s_2^a,t_2^a\},\,\{s_3^a,t_3^a\}\}$ using a supply edge $e\in\{\{w_1,t_3^a\},\,\{s_1^a,w_2\}\}$, then at least $2+\tfrac{1}{2}=\tfrac{5}{2}=u(e)+\tfrac{3}{2}$ units of flow would be routed on $e$.
Therefore, we can assume that each $h\in\{\{s_1^a,t_1^a\},\,\{s_2^a,t_2^a\},\,\{s_3^a,t_3^a\}\}$ is routed in $y$ using only edges of the cycle $C^a=s_1^a,\,s_2^a,\,s_3^a,\,t_1^a,\,t_2^a,\,t_3^a,\,s_1^a$.
Recall that the edge capacities of $C^a$ in $\mathcal{R}^{(a)}$ and in $\mathcal{I}$ differ only on the edge $\{t_3^a,s_1^a\}$.
Thus, by Lemma~\ref{claim:warmup} we may assume that $\{s_1^a,t_1^a\}$ and $\{s_3^a,t_3^a\}$ are routed (counter-clockwise) in $y$ along the paths $Q_1^a=s_1^a,\,t_3^a,\,t_2^a,\,t_1^a$ and $Q_3^a=s_3^a,\,s_2^a,\,s_1^a,\,t_3^a$, respectively; otherwise, the capacity of an edge in $E(C^a)\setminus \{\{t_3^a,s_1^a\}\}$ would be exceeded by at least $\tfrac{5}{4}$.
Observe that if $\{t_3^a,s_1^a\}\in P_{1,3}\cap P_{2,4}$, then
\[y(t_3^a,s_1^a)\ge d(w_1,w_3)+d(w_2,w_4)+d(s_1^a,t_1^a)+d(s_3^a,t_3^a)=4=\tfrac{11}{4}+\tfrac{5}{4}=u(t_3^a,s_1^a)+\tfrac{5}{4}.\]
Therefore, we assume without loss of generality that $\{t_3^a,s_1^a\}\notin P_{1,3}$.
This implies that $E(C^a)\setminus \{\{t_3^a,s_1^a\}\}\subseteq P_{1,3}$.
We argue that $y$ congests either $\{s_1^a,s_2^a\}$ or $\{s_2^a,s_3^a\}$ with at least $\tfrac{5}{4}$ units of flow in that case.
Let $Q_2^a$ denote the path on which $y$ routes $\{s_2^a,t_2^a\}$.
If $\{s_1^a,s^a_2\}\in Q_2^a$, then
\[y(s_1^a,s^a_2)\ge d(w_1,w_3)+d(s_2^a,t_2^a)+d(s_3^a,t_3^a)=\tfrac{5}{2}=u(s_1^a,s_2^a)+\tfrac{5}{4};\]
otherwise, if $\{s_2^a,s^a_3\}\in Q_2^a$, then
\[y(s_2^a,s^a_3)\ge d(w_1,w_3)+d(s_2^a,t_2^a)+d(s_3^a,t_3^a)=\tfrac{5}{2}=u(s_2^a,s_3^a)+\tfrac{5}{4}.\qedhere \]
\end{proof}

\subsection{Proof of Theorem~\ref{theorem:lower}}~\label{section:theoremlower}
We generalize the result in the previous section and show that for any $n\in \Z_{>0}$, there is a feasible multicommodity flow instance supported on an outerplanar graph with $1+4n$ inner faces,
such that any unsplittable flow must necessarily exceed the capacity of some edge by at least 
$\left(1+\tfrac{n}{1+3n}\right)\Dmax$.
Since $1+\tfrac{n}{1+3n}\rightarrow\tfrac{4}{3}$ as $n\rightarrow \infty$, Theorem~\ref{theorem:lower} will follow.\footnote{One can take $n=\lceil\tfrac{1}{9\varepsilon}\rceil$, where $\varepsilon>0$ is as in the statement of Theorem~\ref{theorem:lower}.}

We start by considering a family of ring-loading instances on a common (supply graph, demand graph) pair $(C,H)$.
For any $n\in\Zp$ and any $\ell\in[n]$ we define the following ring-loading instance $\mathcal{I}(n,\ell):=(C,u,H,d)$, which we illustrate in Figure~\ref{fig:i_nl}.

\begin{figure}[h]
    \centering
    \includegraphics[width=0.3\textwidth]{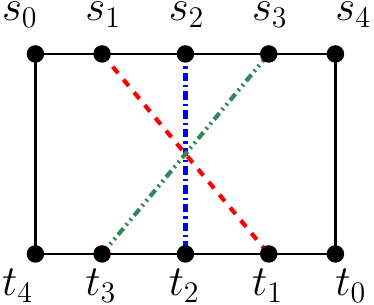}

    \vspace{5mm}
    
    \caption{The ring-loading instance $\Iscr(n,\ell)$}\label{fig:i_nl}
\end{figure}

\begin{align*}
\textbf{Supply graph:}&\quad C=s_0,s_1,s_2,s_3,s_4,t_0,t_1,t_2,t_3,t_4,s_0 \text{ is a 10-cycle},\\[1ex]
\textbf{Demand graph:}&\quad E(H)=\{\{s_1,t_1\},\,\{s_2,t_2\},\,\{s_3,t_3\}\},\\[1ex]
\textbf{Capacities:}&\quad u(s_1,s_2)=u(s_2,s_3)=1+\tfrac{n}{1+3n},\\
&\quad u(s_3,s_4)=u(s_4,t_0)=u(t_0,t_1)=1-\tfrac{\ell}{1+3n},\\
&\quad u(t_1,t_2)=u(t_2,t_3)=1+\tfrac{n}{1+3n},\\
&\quad u(t_3,t_4)=u(t_4,s_0)=u(s_0,s_1)= 1+\tfrac{\ell+2n}{1+3n},\;\\[1ex]
\textbf{Demands:}&\quad d(s_1,t_1)=1,\quad d(s_2,t_2)=\tfrac{2n}{1+3n},\quad d(s_3,t_3)=1.
\end{align*}

\begin{lemma}\label{claim:i_nl_is_feasible}
For any $n\in\Zp$ and $\ell \in[n]$, $\mathcal{I}(n,\ell)$ is a feasible instance.
\end{lemma}
\begin{proof}
For each $i\in\{1,2,3\}$, let $r_i$ denote the amount of $s_i$-$t_i$ flow routed (clockwise) along the path $s_i,s_{i+1},\hdots,s_4,t_0,\hdots,t_i$, and let $\overline{r_i}$ denote the amount of $s_i$-$t_i$ flow routed (counter-clockwise) along the path $s_i,s_{i-1},\hdots,s_0,t_4,\hdots,t_i$.
Take \[r_{1}=r_{3}=\frac{\tfrac{1}{2}+n}{1+3n}-\frac{\ell}{2+6n},\quad \overline{r_{1}}=\overline{r_{3}}=\frac{\tfrac{1}{2}+2n}{1+3n}+\frac{\ell}{2+6n},\quad r_{2}=\overline{r_{2}}=\frac{n}{1+3n}.\]
This flow routes the required demand, since $r_1+\overline r_1=r_3+\overline r_3=1$ and $r_2+\overline r_2=\frac{2n}{1+3n}$.
Observe that the flow on each edge of the paths $s_1,s_2,s_3$ and $t_1,t_2,t_3$ is equal to $1+\frac{n}{1+3n}.$
The flow on each edge of the path $s_3,s_4,t_0,t_1$ is equal to $1-\frac{\ell}{1+3n}$.
Finally, the flow on each edge of the path $t_3,t_4,s_0,s_1$ is equal to $1+\frac{\ell+2n}{1+3n}$.
Thus, on every edge the total flow equals its capacity, and the instance is feasible.
\end{proof}

Next, we construct an (intermediate) instance as follows.
Consider the $n$ vertex-disjoint instances $\{\mathcal{I}(n,\ell)=(C^{(\ell)},\,u^{(\ell)},\,H^{(\ell)},\,d^{(\ell)})\}_{\ell\in[n]}$,
where we assume that $\mathcal{I}(n,\ell)$ is supported on the cycle $C^{(\ell)}=s_0^\ell,\hdots,s_4^\ell,t_0^\ell,\hdots,t_4^\ell,s_0^\ell$ (and where we interpret $s^\ell_i\equiv s_i$ and $t^\ell_i\equiv t_i$).
We concatenate these instances into a larger instance whose inner faces form a one-dimensional grid.
Informally, for every $\ell\in[n-1]$, we glue the edge $\{t_4^{\ell+1},s_0^{\ell+1}\}$ at $\{s_4^\ell,t_0^\ell\}$ by identifying
$s_0^{\ell+1}$ with $s_4^\ell$ and $t_4^{\ell+1}$ with $t_0^\ell$.
Then, we add the capacity of $\{t_4^{\ell+1},s_0^{\ell+1}\}$ to that of $\{s_4^\ell,t_0^\ell\}$.
The demands and the capacities of all other edges (i.e., the outer edges of the resulting instance) remain unchanged.
This results in the following instance $\mathcal{J}=(G,u,H,d)$. See Figure~\ref{fig:ring_loading_sequence} for an illustration.

\begin{figure}[h]
    \centering
    \includegraphics[width=0.8\textwidth]{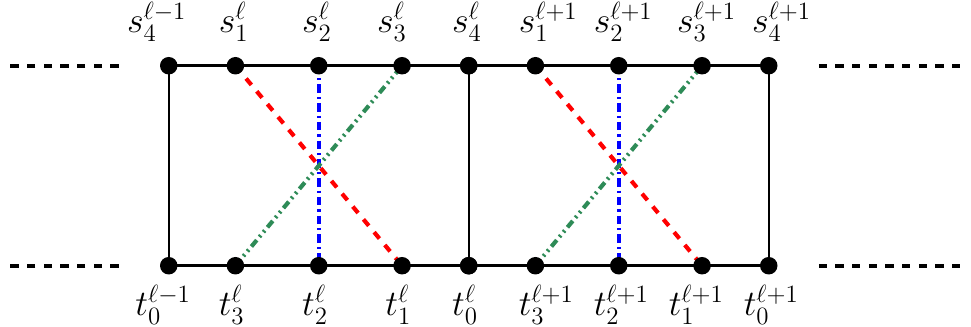}

    \vspace{5mm}
    \caption{The $\ell$-th and $\ell +1$-st inner face in the outerplanar instance $\mathcal{J}$.}\label{fig:ring_loading_sequence}
\end{figure}

\begin{align*}
\textbf{Vertex set:}& \quad V=\bigcup_{\ell=1}^n V(C^{(\ell)})\setminus \bigcup_{\ell=2}^n\{s_0^\ell,t^\ell_4\}.\\
\textbf{Supply edges:}
&\quad E(G)=\bigcup_{\ell=1}^n E(C^{(\ell)})\setminus \{\{t^\ell_4,s_0^\ell\}\,:\;2\le\ell\le n\;\}.\\[1ex]
\textbf{Demand edges:}&\quad E(H)=\bigcup_{\ell=1}^n\{\{s_1^\ell,t_1^\ell\},\,\{s_2^\ell,t_2^\ell\},\,\{s_3^\ell,t_3^\ell\}\}.\\[1ex]
\textbf{Capacities:}&\quad u(s_1^\ell,s^\ell_2)=u(s^\ell_2,s^\ell_3)=u(t^\ell_1,t^\ell_2)=u(t^\ell_2,t^\ell_3)=1+\tfrac{n}{1+3n}\quad \forall\ell\in[n]\;,\\[0.5ex]
&\quad u(s^\ell_3,s^\ell_4)=u(t^\ell_0,t^\ell_1)=1-\tfrac{\ell}{1+3n}\quad\forall\ell\in[n],\\[0.5ex]
&\quad u(t_3^{\ell},t^{\ell-1}_0)=u(s^{\ell-1}_4,s^{\ell}_1)= 1+\tfrac{\ell+2n}{1+3n}\quad \forall2\le\ell\le n,\;\\[0.5ex]
&\quad u(s_4^\ell,t_0^\ell)=u^{(\ell)}(s_4^\ell,t_0^\ell)+ u^{(\ell+1)}(t_4^{\ell+1},s_0^{\ell+1})=3-\tfrac{n}{1+3n}\quad\forall\ell\in[n-1],\\[0.5ex]
&\quad u(t_3^1,t_4^1)=u(t_4^1,s_0^1)=u(s_0^1,s_1^1)=2-\tfrac{n}{1+3n},\quad u(s_4^n,t_0^n)=1-\tfrac{n}{1+3n}.\\[1ex]
\textbf{Demands:}&\quad d(s_1^\ell,t_1^\ell)=1,\quad d(s_2^\ell,t_2^\ell)=\tfrac{2n}{1+3n},\quad d(s_3^\ell,t_3^\ell)=1\quad\forall\ell\in[n].
\end{align*}
Note that $\Dmax=1$ and that removing the outer face from the planar dual of $G$ results in a path whose $\ell^{th}$ node corresponds to the inner face of $G$ bounded by the cycle $C^{(\ell)}$.
\begin{lemma}
$\mathcal{J}$ is a feasible instance.
\end{lemma}
\begin{proof}
 For each $\ell \in [n]$ we can route the demands $\{s_1^\ell, t_1^\ell\},\, \{s_2^\ell, t_2^\ell\},\, \{s_3^\ell, t_3^\ell\}$ as in Lemma~\ref{claim:i_nl_is_feasible},
 around the boundary of the inner face that contains all $s_i^\ell$ and $t_i^\ell$ for $i=1,2,3$.
\end{proof}

For simplicity, we introduce the following definition.
\begin{definition}
Consider an unsplittable flow $y$ of $\mathcal{J}$, and let $\ell\in[n]$ and $i\in\{1,2,3\}$. We say that $y$ routes $\{s_i^\ell,t_i^\ell\}$ \emph{to the right} if $y$ 
routes $\{s_i^\ell,t_i^\ell\}$ along an $s_i^\ell \,t_i^\ell$-path containing the edge $\{s^\ell_i,s^\ell_{i+1}\}$;
otherwise, we say that $y$ routes $\{s_i^\ell,t_i^\ell\}$ \emph{to the left}, i.e., 
if $y$ routes $\{s_i^\ell,t_i^\ell\}$ along an $s_i^\ell \,t_i^\ell$-path containing the edge $\{s^\ell_i,s^\ell_{i-1}\}$.\footnote{We interpret $s^\ell_{0}\equiv s_4^{\ell-1}$ for any $2\le\ell\le n$}
\end{definition}
Later, we construct the instance satisfying the conditions of Theorem~\ref{theorem:lower} by taking four vertex-disjoint copies of $\mathcal{J}$ and identifying the edge $\{t_4^1,s_0^1\}$ of each copy with an edge of a central cycle.
The next lemma will allow us to assume later that, within every copy, all demands $\{s_1^\ell,t_1^\ell\}$ and $\{s_3^\ell,t_3^\ell\}$ must be routed to the left.
\begin{lemma}\label{lemma:totheleft}
Let $y$ be an unsplittable flow of $\mathcal{J}=(G,u,H,d)$.
If, for some $\ell\in[n]$ and some $i\in\{1,3\}$, $y$ routes $\{s_i^\ell,t_i^\ell\}$ to the right, then there is an edge $e\in E(G) \setminus \{\{t_4^1,s_0^1\}\}$ with 
\[y(e)\ge u(e)+1+\tfrac{n}{1+3n}.\]
\end{lemma}
\begin{proof}
We start by observing that if the demands $\{s_1^\ell,t_1^\ell\}$ and $\{s_3^\ell,t_3^\ell\}$ are routed in opposite directions, 
then $y$ exceeds the capacity of some edge of the $\ell$-th face of $G$ by at least $1+\tfrac{n}{1+3n}$.
\begin{claiminproof}\label{claim:opposite}
If for some $\ell\in[n]$ and distinct $i,j\in\{1,3\}$,
$y$ routes $\{s_i^\ell,t_i^\ell\}$ to the right and it routes $\{s_j^\ell,t_j^\ell\}$ to the left, then there is a supply edge $e\in E(G)$ with   
$y(e)\ge u(e)+1+\frac{n}{1+3n}$.
\end{claiminproof}
\noindent\emph{Proof.}
Suppose first that $i=1$ and $j=3$.
Observe that in this case both $\{s_1^\ell,t_1^\ell\}$ and $\{s_3^\ell,t_3^\ell\}$ are routed in $y$ by using the edges of the path $s_1^\ell,s_2^\ell,s_3^\ell$.
Thus, if $\{s_2^\ell,t_2^\ell\}$ is routed to the left, then 
\[y(s_1^\ell,s_2^\ell)\ge d(s_1^\ell,t_1^\ell)+d(s_2^\ell,t_2^\ell)+d(s_3^\ell,t_3^\ell)=2+\tfrac{2n}{1+3n}=2\left(1+\tfrac{n}{1+3n}\right)=u(s_1^\ell,s_2^\ell)+1+\tfrac{n}{1+3n}.\]
On the other hand, if $\{s_2^\ell,t_2^\ell\}$ is routed to the right, then
\[y(s_2^\ell,s_3^\ell)\ge d(s_1^\ell,t_1^\ell)+d(s_2^\ell,t_2^\ell)+d(s_3^\ell,t_3^\ell)=2+\tfrac{2n}{1+3n}=2\left(1+\tfrac{n}{1+3n}\right)=u(s_2^\ell,s_3^\ell)+1+\tfrac{n}{1+3n}.\]
The case in which $i=3$ and $j=1$ is analogous.
In that case, the capacity of either $\{t_1^\ell,t_2^\ell\}$ or $\{t_2^\ell,t_3^\ell\}$ is exceeded by at least $1+\frac{n}{1+3n}$ units of flow.
\hfill{ \Large$\triangleleft$}
\vspace{2ex}

By Claim~\ref{claim:opposite}, we can assume that for all $\ell\in[n]$ either both  $\{s_1^\ell,t_1^\ell\}$ and $\{s_3^\ell,t_3^\ell\}$ are routed to the right, 
or that both are routed to the left.
Take the largest $\ell$ such that both $\{s_1^\ell,t_1^\ell\}$ and $\{s_3^\ell,t_3^\ell\}$ are routed to the right.
If $\ell=n$,
then $y$ routes at least $2=d(s_1^n,t_1^n)+d(s_3^n,t_3^n)$ units of flow through the edge $\{s_4^n,t_0^n\}$. 
Thus,
\[y(s_4^n,t_0^n)\ge 2=u(s_4^n,t_0^n)+1+\tfrac{n}{1+3n}.\]
Therefore, we can assume that $\ell<n$.
For simplicity, let $P_1^\ell$ and $P_3^\ell$ denote the paths on which $y$ routes $\{s_1^\ell,t_1^\ell\}$ and $\{s_3^\ell,t_3^\ell\}$, respectively.
Similarly, let $P_1^{\ell+1}$ and $P_3^{\ell+1}$ denote the paths on which $y$ routes $\{s_1^{\ell+1},t_1^{\ell+1}\}$ and $\{s_3^{\ell+1},t_3^{\ell+1}\}$, respectively.
By our choice of $\ell$ (and the assumption that for each of these $n$ indices, the two corresponding demands are routed in the same direction),
we have that 
\[\{s_3^\ell,s_4^\ell\}\in P_1^{\ell}\cap P_3^{\ell} \;\;\text{ and }\;\;\{s_4^\ell,s_1^{\ell+1}\}\in P_1^{\ell+1}\cap P_3^{\ell+1}.\]
Observe that if these four (unit) demands are routed on the inner edge $\{s_4^\ell,t_0^\ell\}$ i.e., $\{s_4^\ell,t_0^\ell\}\in P_1^{\ell}\cap P_3^{\ell} \cap P_1^{\ell+1}\cap P_3^{\ell+1}$, then
\[y(s_4^\ell,t_0^\ell)-u(s_4^\ell,t_0^\ell)\ge
4-\left(3-\tfrac{n}{1+3n}\right)=1+\tfrac{n}{1+3n}.\]
Thus, we assume that either $\{s_3^\ell,s_4^\ell\}\in P_1^{\ell+1}\cup P_3^{\ell+1}$ or $\{s_4^\ell,s_1^{\ell+1}\}\in P_1^{\ell}\cup P_3^{\ell}$.
If $\{s_3^\ell,s_4^\ell\}\in P_1^{\ell+1}\cup P_3^{\ell+1}$, 
then $y$ would route (at least) $3$ units of flow on $\{s_3^\ell,s_4^\ell\}$, implying that
\[y(s_3^\ell,s_4^\ell)-u(s_3^\ell,s_4^\ell)\ge 3-\left(1-\tfrac{\ell}{1+3n}\right)>2.\]
On the other hand,
if $\{s_4^\ell,s_1^{\ell+1}\}\in P_1^{\ell}\cup P_3^{\ell}$ then $\{s_1^{\ell+1},s_2^{\ell+1}\},\,\{s_2^{\ell+1},s_3^{\ell+1}\}\in P_1^{\ell}\cup P_3^{\ell}$.
Since $\{s_1^{\ell+1},s_2^{\ell+1}\},\,\{s_2^{\ell+1},s_3^{\ell+1}\}\in P_3^{\ell+1}$ as well, this implies that $y$ routes at least $2$ units of flow on these two supply edges without counting the demand $\{s_2^{\ell+1},t_2^{\ell+1}\}$. Thus, the same argument of Claim~\ref{claim:opposite} implies that either
$y(s_1^{\ell+1},s_2^{\ell+1})-u(s_1^{\ell+1},s_2^{\ell+1})\ge 1+\tfrac{n}{1+3n}$
(if $\{s_2^{\ell+1},t_2^{\ell+1}\}$ is routed to the left), or
$y(s_2^{\ell+1},s_3^{\ell+1})-u(s_2^{\ell+1},s_3^{\ell+1})\ge 1+\tfrac{n}{1+3n}$
(if $\{s_2^{\ell+1},t_2^{\ell+1}\}$ is routed to the right).
\end{proof}

Now we describe the instance which will imply the proof of Theorem~\ref{theorem:lower}.
The instance is illustrated in Figure~\ref{fig:lower_bound_4_3}. 
We create four disjoint copies of $\mathcal{J}$ and identify the edge $\{t_4^1,s_0^1\}$ of each of these copies with an edge of a central 12-cycle $W$ containing four new vertices $w_1,w_2,w_3,w_4$ and two ``crossing'' unit demands $\{w_1,w_3\}$, $\{w_2,w_4\}$.
We increase the capacity of the four copies of $\{t_4^1,s_0^1\}$ by one unit, and then assign one unit of capacity to the remaining edges of $W$.
By Lemma~\ref{lemma:totheleft}, we can assume that the demands $\{s_1^1,t_1^1\}$ and $\{s_3^1,t_3^1\}$ of each of these copies are routed to the left; otherwise,
the capacity of some edge within the corresponding copy is exceeded by at least $1+\tfrac{n}{1+3n}$.
Hence, just as in the proof of Theorem~\ref{theorem:weaker}, we can then argue that $\{w_1,w_3\}$, $\{w_2,w_4\}$, together with the demands $\{s_1^1,t_1^1\}$, $\{s_3^1,t_3^1\}$ of one of the four copies of $\mathcal{J}$,
must exceed the capacity of some edge by at least $1+\tfrac{n}{1+3n}$.
We provide a formal proof for completeness.

Consider four vertex-disjoint copies $\{\,\mathcal{J}^{(\theta)}=(G^\theta,\,u^\theta,\,H^\theta,\,d^\theta)\;\}_{\theta\in\{a,b,c,r\}}$ of $\mathcal{J}$.
To simplify notation, we use $t^\theta$ and $s^\theta$ to denote the copies of nodes $t_4^1$ and $s^1_0$ in $\mathcal{J}^{(\theta)}$, respectively. Consider the 12-cycle
\[W=w_1,\,t^a,\,s^a,\,w_2,\,t^b,\,s^b,\,w_3,\,t^c,\,s^c,\,w_4,\,t^r,\,s^r,\,w_1\;.\]
Consider the following instance $\mathcal{I}=(G,u,H,d)$.
\begin{align*}
\textbf{Vertex set:}&\quad V=\{w_1,w_2,w_3,w_4\}\cup \bigcup_{\theta\in\{a,b,c,r\}} V(G^{\theta}).\\
\textbf{Supply edges:}&\quad E(G)= E(W)\cup \bigcup_{\theta\in\{a,b,c,r\}} \big(E(G^{\theta})\setminus \{\{t^\theta,s^\theta\}\}\big).\\[1ex]
\textbf{Demand edges:}&\quad E(H)=\{\{w_1,w_3\},\;\{w_2,w_4\}\} \cup \bigcup_{\theta\in\{a,b,c,r\}}E(H^\theta).\\[1ex]
\textbf{Capacities:} &\quad u(e)=u^\theta(e),\quad \quad\forall \theta\in\{a,b,c,r\},\;\forall e\in E(G^\theta)\setminus\{\{t^\theta,s^\theta\}\}.\\[1ex]
&\quad u(t^\theta,s^\theta)=1+u^\theta(t^\theta,s^\theta)=3-\tfrac{n}{1+3n},\;\qquad \qquad\forall\theta\in\{a,b,c,r\}.\\[1ex]
&\quad u(e)=1,\quad\quad \quad\forall e\in E(W)\setminus\{\{t^\theta,s^\theta\}\;:\;\theta=a,b,c,r\}.\\[1ex]
\textbf{Demands:}
&\quad d(h)=d^\theta(h),\quad\quad \forall \theta\in\{a,b,c,r\},\quad\forall h\in E(H^\theta).\\[1ex]
&\quad d(w_1,w_3)=1,\quad d(w_2,w_4)=1.
\end{align*}

\begin{figure}[h]
    \centering
    \includegraphics[width=0.65\textwidth]{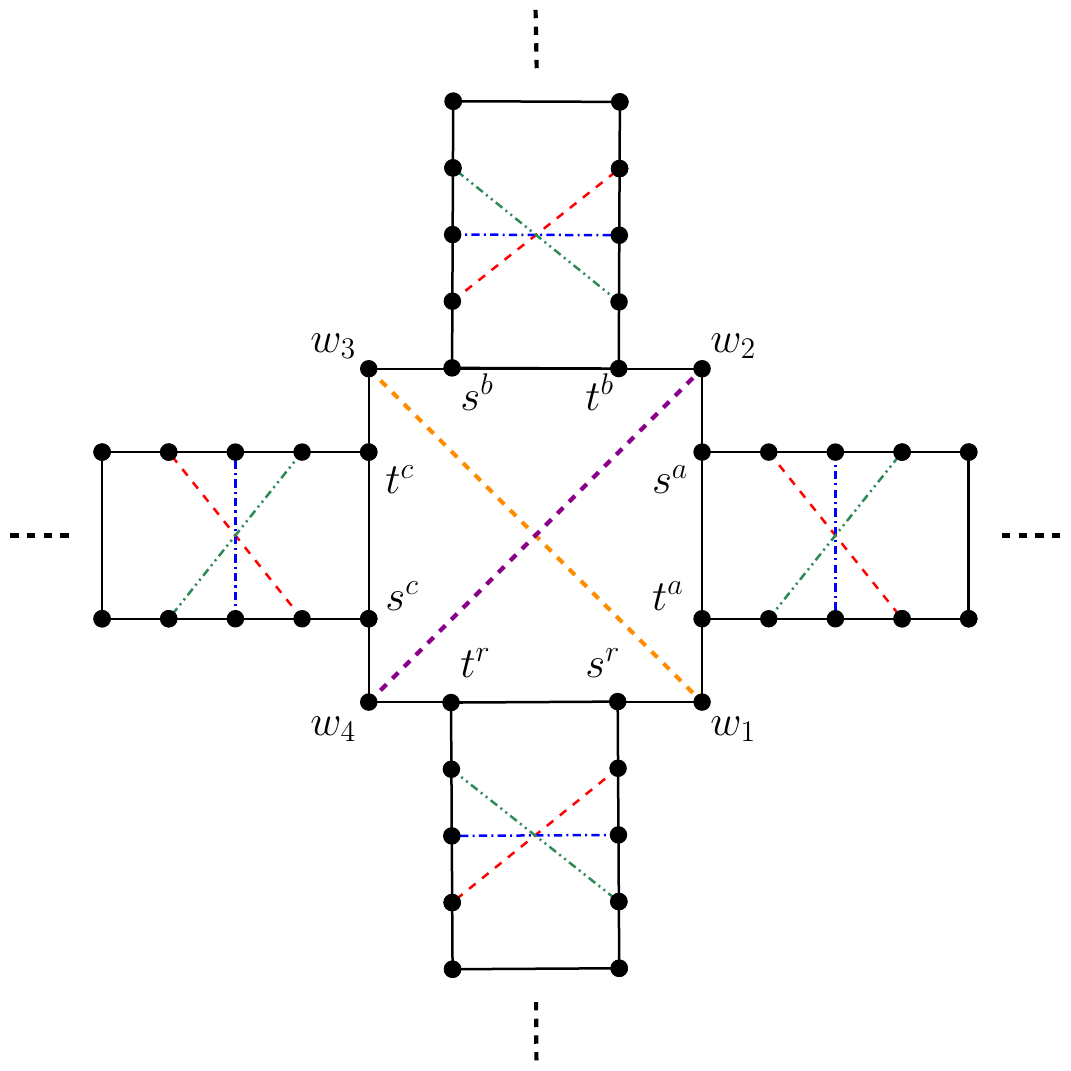}

    \vspace{5mm}
    \caption{The instance $\Iscr$ that consists of four copies of $\mathcal{J}$, joined at the cycle $W$}\label{fig:lower_bound_4_3}
\end{figure}

\begin{lemma}
$\mathcal{I}=(G,u,H,d)$ is a feasible instance.
\end{lemma}
\begin{proof}
First, route $\{w_1,w_3\}$ (resp. $\{w_2,w_4\}$) by sending $\tfrac{1}{2}$ units
of flow along each of the two $w_1w_3$-paths (resp. $w_2w_4$-paths) in $W$.
The residual capacity, with respect to $u$, on each edge
$e\in E(G^\theta)\subseteq E(G)$ is precisely $u^\theta(e)$, for each
$\theta\in\{a,b,c,r\}$. Since each instance $\mathcal{J}^\theta$ is feasible,
the remaining demands can be routed within the corresponding copies. It follows
that $\mathcal{I}$ is feasible.
\end{proof}
We can now prove Theorem~\ref{theorem:lower}.
\begin{proof}[Proof of Theorem~\ref{theorem:lower}]
As mentioned at the start of the section,
it suffices to show that for any unsplittable flow $y$ of $\mathcal{I}$,
there exists some $e\in E(G)$ with $y(e)\ge u(e)+1+\tfrac{n}{1+3n}$ (note that $\Dmax=1$).
Let $P_{1,3}$ and $P_{2,4}$ denote the paths on which $y$ routes $\{w_1,w_3\}$ and $\{w_2,w_4\}$, respectively.
Then, there exists some $i\in[4]$ and $\theta\in\{a,b,c,r\}$ such that $\{w_i,t^\theta\},\{s^\theta,w_{i+1}\}\in P_{1,3}\cap P_{2,4}$ (where we interpret $w_5\equiv w_1$).
Suppose without loss of generality (by the symmetry of the instance) that $\{w_1,t^a\},\{s^a,w_2\}\in P_{1,3}\cap P_{2,4}$.
Observe that if $y$ routes any demand of $E(H^a)$ on some $e\in \{\{w_1,t^a\},\{s^a,w_2\}\}$, then
\[y(e)\ge d(w_1,w_3)+d(w_2,w_4)+\tfrac{2n}{1+3n}=2+\tfrac{2n}{1+3n}=u(e)+1+\tfrac{2n}{1+3n},\]
where we use the fact that $\frac{2n}{1+3n}$ is the \emph{minimum} demand value in $\mathcal{I}$.
Thus, we can assume that $y$ routes every demand in $E(H^a)$ using only the edges of $E(G^a)$.
We overload notation and use $s^1_i$ and $t^1_i$ for each $i\in\{0,1,2,3,4\}$ to denote the corresponding copies of these nodes of $\mathcal{J}$ in $V(G^a)$ (note that $s^1_0\equiv s^a$ and $t^1_4\equiv t^a$).
Since $u(e)=u^a(e)$ for every $e\in E(G^a)\setminus\{\{t^a,s^a\}\}$, 
Lemma~\ref{lemma:totheleft} allows us to assume that, in $\mathcal{J}^a$, all demands $\{s_1^\ell,t_1^\ell\}$ and $\{s_3^\ell,t_3^\ell\}$ are routed to the left.
In particular, this holds for $\{s^1_1,t^1_1\}$ and $\{s^1_3,t^1_3\}$.
Let $Q_1$ and $Q_3$ denote the paths on which $y$ routes $\{s^1_1,t^1_1\}$ and $\{s^1_3,t^1_3\}$, respectively. 
By the above, $\{t^a,s^a\}\equiv\{t^1_4,s^1_0\}$ is contained in $Q_1\cap Q_3$.
Observe that if $\{t^a,s^a\}$ is contained in $P_{1,3}\cap P_{2,4}$, then 
\[y(t^a,s^a)\ge d(w_1,w_3)+d(w_2,w_4)+d(s^1_1,t^1_1)+d(s^1_3,t^1_3)=4=u(t^a,s^a)+1+\tfrac{n}{1+3n}.\]
Thus, we can w.l.o.g.\@ assume that $\{s^a,s^1_1\}\equiv\{s^1_0,s^1_1\}$ is contained in $P_{1,3}$ (since $s^a$ has degree $3$). Since $y$ routes $\{s^1_1,t^1_1\}$ and $\{s^1_3,t^1_3\}$ to the left, $\{s^a,s^1_1\}\in Q_1\cap Q_3$. It follows that
\[y(s^a,s^1_1)\ge d(w_1,w_3)+d(s^1_1,t^1_1)+d(s^1_3,t^1_3)=3=u(s^a,s_1^1)+1+\tfrac{n}{1+3n}. \qedhere
\]
\end{proof}

\section*{Acknowledgements}
We thank the anonymous reviewers for their valuable comments and suggestions.
We are grateful to Chaitanya Swamy for detailed feedback on a preliminary version
of this paper, and to Nikhil Kumar, Joseph Poremba, and Bruce Shepherd for many interesting discussions.
This project started while David was visiting the Research Institute for Discrete Mathematics at the University of Bonn, with financial support from the Hausdorff Center for Mathematics, funded by the Deutsche Forschungsgemeinschaft (DFG, German Research Foundation) under Germany's Excellence Strategy — EXC-2047/2 — 390685813. David thanks both institutions for their hospitality during this visit.

\printbibliography
\appendix

\section{Omitted Proofs from Section~\ref{section:2dmax}}\label{appendix:omitted}

We provide a proof for Theorem~\ref{theorem:ring_loading_3_2} and Lemma~\ref{lemma:ring_loading_with_two_demands}.
We first note the following structural result by Schrijver, Seymour, and Winkler~\cite{schrijver1998ringloading}.

\begin{definition}
    Let $(G,u,H,d)$ be a feasible ring-loading instance.
    We say that two demand edges $h_1=\{s_1,t_1\}$ and $h_2=\{s_2,t_2\}$ \emph{cross} if $h_1 \cap h_2 = \emptyset$ and for $i=1,2$ the vertices $s_i$ and $t_i$ lie in different connected components of $G-\{s_{3-i}, t_{3-i}\}$.
\end{definition}

\begin{lemma}[Schrijver, Seymour, Winkler~\cite{schrijver1998ringloading}]\label{lemma:ring_loading_sturcture}
    Let $(G,u,H,d)$ be a feasible ring-loading instance.
    Then there exists a feasible fractional flow $x = (x_h)_{h \in E(H)}$ such that all demands that are not routed unsplittably pairwise \emph{cross}, i.e.,\@ if $h_1 \neq h_2 \in E(H)$ such that $x_{h_i}$ is not unsplittable for $i=1,2$ then $h_1$ and $h_2$ cross.
\end{lemma}

In~\cite{schrijver1998ringloading}, the authors use this observation to show how to compute a $\frac{3}{2}$-feasible unsplittable flow in any feasible ring-loading instance.
Their approach directly implies a stronger guarantee of $\frac{1}{2}$ for a single edge $\etight \in E(G)$, which can be chosen arbitrarily in advance.
We restate Theorem~\ref{theorem:ring_loading_3_2} for convenience.
The proof follows the same argument as in~\cite{schrijver1998ringloading}.

\theoremring*
\begin{proof}
    Let $x = (x_h)_{h \in E(H)}$ be a feasible flow as given by Lemma~\ref{lemma:ring_loading_sturcture}.
    Since $x\le u$, we may assume that $u(e)=x(e)$ for every $e\in E(G)$.
    Let $D_1\subseteq E(H)$ be the set of demands that are routed unsplittably in $x$, and let $D_2:=E(H)\setminus D_1$.
    Observe that any $h_1\neq h_2 \in D_2$ cross. In particular, the demand edges $D_2$ are pairwise vertex-disjoint. 
    Note that we can remove the demands in $D_1$ from the instance along with their unsplittable $h$-flows without making our problem harder i.e., it suffices to prove the theorem for the ``residual'' instance with demands $D_2$, and where the capacity of every $e\in E(G)$ is given by $\sum_{h\in D_2}x_h(e)= u(e)-\sum_{h\in D_1}x_h(e)$.
    Thus, we assume from now on that $E(H)=D_2$.
    
    Additionally, observe that we can assume that every vertex of $G$ is incident to some $h\in E(H)$;
    otherwise, if $P\subseteq E(G)$ is the edge set of a (maximal) path such that none of its inner vertices are incident to any demand $h\in E(H)$, then $u(e_1)=u(e_2)$ for every $e_1,e_2\in P$.
    Thus, we may contract $P$ into an edge $e$ of capacity $u(e):=u(e')$ for any $e'\in P$, and obtain an equivalent (feasible) instance.
    By the above, we may assume that 
    \[E(H)=\{\{s_i,t_i\}:1\le i\le k\},\]
    and that the supply graph is given by the cycle 
    \[G=s_1,s_2,\hdots,s_k,t_1,t_2,\hdots,t_k,s_1.\]
    In particular, we may assume w.l.o.g.\@ that $\etight=\{s_k,t_1\}.$
    
    For each $ i \in [k]$, let $r_i$ denote the amount of flow that $x_{\{s_i,t_i\}}$ routes (clockwise) along the path $s_i,s_{i+1},\hdots,s_k,t_1,\hdots,t_i$, and let $\overline{r_i}=d(h_i)-r_i$ denote the amount of flow that $x_{\{s_i,t_i\}}$ routes (counter-clockwise) along the path $s_i,s_{i-1}\hdots,s_1,t_k,\hdots,t_i$.
    
    Schrijver, Seymour, and Winkler~\cite{schrijver1998ringloading} observed that there is a correspondence between unsplittable flows $y$ and vectors $\vec z \in \R^k$ such that $z_i \in \{-r_i, \,\overline{r_i}\}$ for $i \in [k]$.
    Under this correspondence, for every $j\in[k]$ and every supply edge $e\in\{\{s_j,s_{j+1}\},\,\{t_j,t_{j+1}\}\}$, where we interpret $s_{k+1}\equiv t_1$ and $t_{k+1}\equiv s_1$, the flow difference between $y$ and $x$ on $e$ is $|y(e) - x(e)| = |\sum_{i=1}^j z_i - \sum_{i=j+1}^k z_i|$.
    In particular, $|y(\etight) - x(\etight)| = |\sum_{i=1}^k z_i|$.
    
    The algorithm from~\cite{schrijver1998ringloading} works by iteratively choosing $z_j \in \{-r_j, \overline{r_j}\}$ from $j=1$ to $j=k$ such that $|\sum_{i=1}^j z_i| \leq \frac{\Dmax}{2}$.
    This can be done greedily because $r_j + \overline{r_j}=d(s_j,t_j) \leq \Dmax$.
    As a direct consequence, for each $ j \in [k]$ the triangle inequality yields
    \[\Big|\sum_{i=1}^j z_i - \sum_{i=j+1}^k z_i\Big| 
    =\Big|2\sum_{i=1}^j z_i - \sum_{i=1}^k z_i\Big| 
    \leq 2 \cdot \Big|\sum_{i=1}^j z_i\Big| + \Big|\sum_{i=1}^k z_i\Big| \leq 3 \cdot \frac{\Dmax}{2}.\]
    This shows that the corresponding unsplittable flow satisfies our capacity constraints for edges $e \in E(G) \setminus \{\etight\}$, while for $\etight$ we use $|y(\etight) - x(\etight)| = |\sum_{i=1}^k z_i| \leq \frac{\Dmax}{2}$.
\end{proof}

We now restate and prove Lemma~\ref{lemma:ring_loading_with_two_demands} by a simple case distinction.

\lemmaring*
\begin{proof}
    As in the proof of Theorem~\ref{theorem:ring_loading_3_2}, using Lemma~\ref{lemma:ring_loading_sturcture} and removing the demands already routed unsplittably, we may assume that all demand edges pairwise cross.
    Since every demand edge contains $v$ or $w$, and crossing demand edges are pairwise vertex-disjoint, it follows that $|E(H)| \leq 2$.
    
    Let $x = (x_h)_{h \in E(H)}$ be a feasible flow for $\Iscr$.
    If $|E(H)| = 1$, then we route the unique demand $h$ along the path on which $x$ routes at least  $\tfrac{d(h)}{2}$ units of flow.
    This yields a $\frac{1}{2}$-feasible unsplittable flow for $\Iscr$.
    
    Now consider the case $E(H) = \{\{s_1,t_1\}, \,\{s_2,t_2\}\}$.
    If there is an $s_it_i$-path $P$ in $G-\{e_1, e_2\}$ for some $i \in \{1,2\}$ then we route $\{s_i,t_i\}$ along $P$,
    and route the other demand $\{s_{3-i},t_{3-i}\}$ along the path on which $x_{\{s_{3-i},\,t_{3-i}\}}$ routes at least $\tfrac{d(s_{3-i},\,t_{3-i})}{2}$ units of flow.
    This yields a $\vec \beta$-feasible unsplittable flow.    
    Otherwise, $e_1$ and $e_2$ lie on different $s_it_i$-paths in $G$ for $i=1,2$.
    In particular, any unsplittable flow $y$ must satisfy $y(e_1) - x(e_1) = x(e_2) - y(e_2)$.
    The result then follows by applying Theorem~\ref{theorem:ring_loading_3_2} with $\etight:=e_1$.
    \end{proof}

\end{document}